\documentclass[PDFlatex,sn-mathphys-ay]{sn-jnl}
\usepackage{graphicx}
\usepackage{multirow}
\usepackage{amsmath,amssymb,amssymb, bm, amsfonts}
\usepackage{amsthm}
 \UseRawInputEncoding
\usepackage{natbib}
\usepackage{mathrsfs}
\usepackage[title]{appendix}
\usepackage{xcolor}
\usepackage{float}
\usepackage{placeins}
\usepackage{textcomp}
\usepackage{manyfoot}
\usepackage{booktabs}
\usepackage{algorithm}
\usepackage{algorithmicx}
\usepackage{algpseudocode}
\usepackage{listings}%

\theoremstyle{thmstyleone}%
\newtheorem{theorem}{Theorem}
\newtheorem{Proposition}[theorem]{Proposition}
\theoremstyle{thmstyletwo}

\newtheorem{remark}{Remark}

\theoremstyle{thmstylethree}%

\begin{document}

\title[NTLE]{Transmuted logistic-exponential distribution - some new properties, estimation methods and application with infectious disease mortality data.}

\author*[1,2]{\fnm{Isqeel} \sur{Ogunsola}}\email{isqeel.ogunsola@manchester.ac.uk}

\author[2]{\fnm{Abosede} \sur{Akintunde }}\email{akintundeaa@funaab.edu.ng}

\author[2]{\fnm{ Kehinde} \sur{Yusuff }}\email{yusuffkm@funaab.edu.ng} 

\author[2]{\fnm{Basirat} \sur{Adetona}}\email{sowolabo@funaab.edu.ng}

\author[2]{\fnm{and Faheez} \sur{Abdulrasaq. }}\email{abdulrasaqfaheez3@gmail.com}

\affil*[1]{\orgdiv{Department of Mathematics}, \orgname{University of Manchester}, {\country{United Kingdom}}}

\affil[2]{\orgdiv{Department of Statistics}, \orgname{Federal University of Agriculture, Abeokuta}, {\country{Nigeria}}}

\abstract{Lately, a New Transmuted Logistic-exponential (NTLE) distribution was introduced and studied as an extension of the Logistic-Exponential Distribution (LED) with wider applicability in lifetime modelling.  However, the maximum likelihood estimates (MLE) of NTLE are not in closed form, and the consistency of the estimates was not examined. Furthermore, some other important properties of NTLE, namely the Shannon entropy, R\'enyi entropy, stochastic ordering, mode, stress-strength reliability measure, residual life functions (mean and reverse), incomplete moments, Bonferroni and Lorenz curves are yet to be derived. Motivated by this, we derived and studied these important properties and evaluated the performance of ten estimation methods (Maximum Likelihood, Moments, Least Squares, Weighted Least Squares, Maximum product of Spacings, Anderson-Darling, Cramer-von Mises, percentile estimation, and Maximum Goodness-of-Fit methods) for NTLE parameters via Monte Carlo simulation using bias, mean square error, and root mean square error as evaluation criteria. Real-life infectious mortality data fitted to the distributions showed that NTLE has a better fit compared to its base distributions (Exponential and Logistic-Exponential). This finding contributes valuable insights for researchers and practitioners when selecting the appropriate estimation methods, especially for NTLE and some similar distributions in non-closed form.
}


\maketitle

\section{Introduction}

Classical distributions such as the exponential and logistic models are widely used in reliability, survival, and lifetime data analysis because of their tractability and interpretability \citep{ogunsola2026new}. However, real data often exhibit compound-complex behaviour that cannot always be captured appropriately by simple models. Consequently, the development of more adaptable probability distributions has become an important area of research in statistics today.

Logistic–exponential (LE) distribution combines structural features of the logistic and exponential distributions, producing a more flexible model which has a wide variety of hazard rate behaviours \citep{lan2008logistic}. Due to its flexibility, the LE distribution has been applied in modelling lifetime data with skewed data and complex hazard functions \citep{hajahmad2024ule, dutta2024bayesian}. Despite its usefulness, the logistic–exponential distribution still lacks sufficient flexibility in certain situations involving heavy-tailed or highly skewed data \citep{adesegun2023tle}. To address such limitations, several extensions of the LE distribution have been proposed in recent years \citep{genc2024half,alharbi2025gostll, mansoor2019marshall, adesegun2023tle}. Of interest in this study is the new transmuted logistic-exponential (NTLE) distribution, an extension proposed by \cite{adesegun2023tle} for modelling lifetime data, which has been demonstrated to be useful in reliability applications. The cumulative distribution function (CDF) and the corresponding probability density function (PDF) of NTLE are given, respectively as:

\begin{equation}
G(y)=
\frac{(e^{\lambda y}-1)^\beta\left[1+\delta+(e^{\lambda y}-1)^\beta\right]}
{\left[1+(e^{\lambda y}-1)^\beta\right]^2}, \quad y>0.
\label{equation CDF}
\end{equation}
and 
\begin{equation}
g(y)=
\frac{\beta\lambda e^{\lambda y}(e^{\lambda y}-1)^{\beta-1}
\left[(1+(e^{\lambda y}-1)^\beta)+\delta(1-(e^{\lambda y}-1)^\beta)\right]}
{\left(1+(e^{\lambda y}-1)^\beta\right)^3}.
\label{equation PDF}
\end{equation}

also, the hazard and survival functions are given respectively as:
$$
h(y)=
\frac{
\beta\lambda e^{\lambda y}(e^{\lambda y}-1)^{\beta-1}
\left[1+\delta+(1-\delta)(e^{\lambda y}-1)^\beta\right]
}{
\left[1+(e^{\lambda y}-1)^\beta\right]
\left[1+(1-\delta)(e^{\lambda y}-1)^\beta\right]
},
$$
and 
$$(y) = \frac{{1 + (e^{\lambda y}  - 1)^\beta  (1 - \delta )}}{{[1 + (e^{\lambda y}  - 1)^\beta  ]^2 }}$$ 
See \citep{adesegun2023tle} for details about NTLE distribution. However, many important properties of NTLE are yet to be studied. This is one of the gaps among others that this study tends to fill.  

Alongside the development of new probability models, the problem of parameter estimation has received significant attention in statistical research \citep{alsadat2024tankedset, habib2024tnhr}. Reliable estimation procedures are essential for determining whether a proposed distribution can adequately represent observed data. Maximum likelihood estimation (MLE) \citep{Fisher1922, casellaBerger2002} remains the most widely used method because of its desirable asymptotic properties, including consistency and efficiency. Nevertheless, likelihood-based estimators may sometimes exhibit instability when the likelihood equations are difficult to solve numerically or when sample sizes are small. For this reason, many studies investigating new probability distributions compare several estimation techniques in order to evaluate their relative performance \citep{ nombebe2023fitting, alabdulhadi2024different}. Recent works on newly proposed lifetime distributions frequently compare MLE with alternative estimators such as least squares, weighted least squares, Anderson–Darling, Cramér–von Mises, and maximum product of spacings estimators \citep{habib2024tnhr,benchiha2025lindley}. 

Minimum-distance estimators such as the least squares (LSE) and weighted least squares (WLSE) estimators are commonly employed when fitting theoretical distribution functions to empirical data \citep{Swain1988}. Several recent studies have applied LSE and WLSE estimators in parameter estimation problems involving newly proposed probability models \citep{pramanik2023xgamma,onuoha2023weibull}. Goodness-of-fit based estimators such as the Anderson–Darling and Cramér–von Mises estimators have also become increasingly popular in the statistical literature \citep{andersonDarling1952, cramer1928}. Recent research has demonstrated that Anderson–Darling and cramér–von Mises estimators perform well in parameter estimation studies involving newly proposed lifetime distributions \citep{habib2024tnhr,  pramanik2023xgamma}. 

Another estimation technique that has attracted attention in recent studies is the maximum product of spacings (MPS) estimator \citep{Luceno2006}. Several simulation studies have shown that MPS estimators often perform comparably to maximum likelihood estimators, particularly when dealing with moderate sample sizes or complex likelihood functions \citep{cheng1979maximum}. Bayesian estimation methods have also gained prominence in recent statistical research. Modern implementations of Bayesian inference typically rely on Markov chain Monte Carlo algorithms, which allow posterior distributions to be approximated efficiently even when analytical solutions are not available \citep{gelman2014bayesian}. 

Although numerous studies have introduced extensions of logistic-type distributions and examined various estimation techniques \citep{mansoor2019marshall,ali2020two}, comprehensive investigations involving a wide range of estimation methods for recently introduced transmuted logistic-exponential distributions parameters are not in existence. In particular, no study has compared likelihood-based, moment-based, goodness-of-fit, and Bayesian estimators simultaneously for the new transmuted logistic-exponential type model (NTLE). 

Motivated by these gaps in the literature, the present study derived some new important properties of NTLE, evaluated and compared methods for its parameter estimation and examined its improvement in modelling skewed distributions compared to its base distributions (exponential and logistic). Specifically, ten estimation methods are considered: maximum likelihood estimation, method of moments estimation, least squares estimation, weighted least squares estimation, maximum product of spacings estimation, Bayesian estimation, Anderson–Darling estimation, Cramér–von Mises estimation, percentile estimation, and maximum goodness-of-fit estimation. A comprehensive Monte Carlo simulation study is conducted to evaluate the performance of these estimators under different sample sizes and parameter values. The findings provide insight into the relative efficiency and robustness of the competing estimation procedures and demonstrate the potential usefulness of the NTLE distribution for modelling lifetime data in practical applications.

\section{Some New properties of NTLE}

In this section, we present several useful analytical properties of the New Transmuted Logistic exponential (NTLE) distribution not previously studied, namely the Shannon entropy, R\'enyi entropy, a stochastic ordering, mode, stress-strength, residual life functions (mean and reverse), incomplete moments, Bonferroni and Lorenz curves. These properties provide insight into its behaviour and practical usefulness in real-world applications.

\subsection{Shannon entropy}

Shannon entropy is one of the most widely used measures of uncertainty in a probability distribution. In lifetime and reliability analysis, Shannon entropy is useful for comparing competing distributions in terms of dispersion and uncertainty \citep{kayid2023entropy, machado2021entropy}. A larger entropy value generally indicates greater unpredictability in the underlying lifetime mechanism.

For a continuous random variable $Y$ with density function $g(y)$, the Shannon entropy is defined by
\begin{equation}
H(Y)=-\int_{0}^{\infty} g(y)\log g(y)\,dy.
\end{equation}

\begin{theorem} 
Let $Y\sim NTLE(\lambda, \beta, \delta)$ with probability density function 
\begin{equation*}
g(y)=
\frac{\beta\lambda e^{\lambda y}(e^{\lambda y}-1)^{\beta-1}}
{\left[1+(e^{\lambda y}-1)^{\beta}\right]^3}
\left[1+\delta+(1-\delta)(e^{\lambda y}-1)^{\beta}\right],
\qquad y>0,
\end{equation*}
where $\lambda>0$, $\beta>0$ and $-1<\delta<1$. Then, the Shannon entropy is in the form
\begin{equation}
H(Y)
=
2-\log(\beta\lambda)-\frac{\delta}{\beta}-j_{\beta,\delta}-K_{\delta},
\label{eq:shannon_main_NTLE}
\end{equation}
where
\begin{equation}
j_{\beta,\delta}
=
\int_{0}^{1}
(1+\delta-2\delta v)
\log\!\left[
1+\left(\frac{v}{1-v}\right)^{1/\beta}
\right]\,dv,
\label{eq:jbd_NTLE}
\end{equation}
and
\begin{equation}
K_{\delta}
=
\int_{0}^{1}
(1+\delta-2\delta v)\log(1+\delta-2\delta v)\,dv.
\label{eq:Kd_int_NTLE}
\end{equation}
For $\delta\neq 0$, the term $K_{\delta}$ has the closed form
\begin{equation}
K_{\delta}
=
\frac{(1+\delta)^2\log(1+\delta)-(1-\delta)^2\log(1-\delta)}{4\delta}
-\frac{1}{2},
\label{eq:Kd_closed_NTLE}
\end{equation}
while, by implication,
\begin{equation*}
K_{0}=0.
\end{equation*}
as a result, for $\delta\neq 0$,
\begin{equation}
H(Y)
=
\frac{5}{2}
-\log(\beta\lambda)
-\frac{\delta}{\beta}
-j_{\beta,\delta}
-
\frac{(1+\delta)^2\log(1+\delta)-(1-\delta)^2\log(1-\delta)}{4\delta}.
\label{eq:shannon_closed_NTLE}
\end{equation}
\end{theorem}

\begin{proof}
By definition,
\begin{equation*}
H(Y)=-\int_{0}^{\infty} g(y)\log g(y)\,dy.
\end{equation*}
Let, 
\begin{equation*}
u=(e^{\lambda y}-1)^\beta.
\end{equation*}
Then, we have; 
\begin{equation*}
e^{\lambda y}=1+u^{1/\beta},
\qquad
du=\beta\lambda e^{\lambda y}(e^{\lambda y}-1)^{\beta-1}dy,
\end{equation*}
so that
\begin{equation*}
g(y)\,dy
=
\frac{1+\delta+(1-\delta)u}{(1+u)^3}\,du.
\label{eq:gdy_u_NTLE}
\end{equation*}
and,
\begin{equation}
g(y)
=
\frac{\beta\lambda(1+u^{1/\beta})u^{(\beta-1)/\beta}\,[1+\delta+(1-\delta)u]}
{(1+u)^3}.
\end{equation}
Hence
\begin{align}
H(Y)
&=
-\int_{0}^{\infty}
\frac{1+\delta+(1-\delta)u}{(1+u)^3}
\log\!\left[
\frac{\beta\lambda(1+u^{1/\beta})u^{(\beta-1)/\beta}[1+\delta+(1-\delta)u]}
{(1+u)^3}
\right]du
\nonumber\\
&=
-\log(\beta\lambda)
-I_{1}
-\frac{\beta-1}{\beta}I_{2}
-I_{3}
+3I_{4},
\label{eq:entropy_split_NTLE}
\end{align}
where
\begin{align*}
I_{1}
&=
\int_{0}^{\infty}
\frac{1+\delta+(1-\delta)u}{(1+u)^3}
\log(1+u^{1/\beta})\,du,
\\
I_{2}
&=
\int_{0}^{\infty}
\frac{1+\delta+(1-\delta)u}{(1+u)^3}
\log u\,du,
\\
I_{3}
&=
\int_{0}^{\infty}
\frac{1+\delta+(1-\delta)u}{(1+u)^3}
\log\!\bigl(1+\delta+(1-\delta)u\bigr)\,du,
\\
I_{4}
&=
\int_{0}^{\infty}
\frac{1+\delta+(1-\delta)u}{(1+u)^3}
\log(1+u)\,du.
\end{align*}

Using the transformation,
\begin{equation}
v=\frac{u}{1+u},
\qquad
u=\frac{v}{1-v},
\qquad
du=\frac{dv}{(1-v)^2},
\qquad 0<v<1.
\end{equation}
Simplifying, we obtain
\begin{equation}
\frac{1+\delta+(1-\delta)u}{(1+u)^3}\,du
=
(1+\delta-2\delta v)\,dv.
\label{eq:weight_v_NTLE}
\end{equation}
Evaluating $I_{1}, I_{2}, I_{3}, I_{4}$, we have; 

For $I_{2}$, we have 
\[
\log u=\log v-\log(1-v),
\]

\begin{equation*}
I_{2}
=
\int_{0}^{1}(1+\delta-2\delta v)\log v\,dv
-
\int_{0}^{1}(1+\delta-2\delta v)\log(1-v)\,dv.
\end{equation*}
Using the standard integrals
\[
\int_{0}^{1}\log v\,dv=-1,
\qquad
\int_{0}^{1}v\log v\,dv=-\frac14,
\qquad
\int_{0}^{1}\log(1-v)\,dv=-1,
\qquad
\int_{0}^{1}v\log(1-v)\,dv=-\frac34,
\]
we obtain
\begin{equation*}
I_{2}=-\delta.
\label{eq:i2_NTLE}
\end{equation*}

also, for $I_{4}$,
\[
\log(1+u)=-\log(1-v),
\]
and
\begin{equation*}
I_{4}
=
-\int_{0}^{1}(1+\delta-2\delta v)\log(1-v)\,dv
=
1-\frac{\delta}{2}.
\label{eq:i4_NTLE}
\end{equation*}

For $I_{3}$, note that
\[
1+\delta+(1-\delta)u
=
\frac{1+\delta-2\delta v}{1-v},
\]
Hence, 
\begin{align*}
I_{3}
&=
\int_{0}^{1}(1+\delta-2\delta v)
\log(1+\delta-2\delta v)\,dv
-
\int_{0}^{1}(1+\delta-2\delta v)\log(1-v)\,dv
\nonumber\\
&=
K_{\delta}
+
1-\frac{\delta}{2},
\label{eq:i3_NTLE}
\end{align*}
where $K_{\delta}$ is given by \eqref{eq:Kd_int_NTLE}.

For $\delta\neq 0$,
let $t=1+\delta-2\delta v$. Then $dt=-2\delta\,dv$, and
\begin{align}
K_{\delta}
&=
\int_{0}^{1}(1+\delta-2\delta v)\log(1+\delta-2\delta v)\,dv
\nonumber\\
&=
\frac{1}{2\delta}\int_{1-\delta}^{1+\delta} t\log t\,dt
\nonumber\\
&=
\frac{1}{2\delta}
\left[
\frac{t^{2}}{2}\log t-\frac{t^{2}}{4}
\right]_{1-\delta}^{1+\delta}
\nonumber\\
&=
\frac{(1+\delta)^2\log(1+\delta)-(1-\delta)^2\log(1-\delta)}{4\delta}
-\frac12.
\end{align}
This proves \eqref{eq:Kd_closed_NTLE}. The case $\delta=0$ follows by continuity, giving $K_{0}=0$.

Finally, from \eqref{eq:weight_v_NTLE},
\begin{equation}
I_{1}
=
\int_{0}^{1}
(1+\delta-2\delta v)
\log\!\left[
1+\left(\frac{v}{1-v}\right)^{1/\beta}
\right]dv
=
j_{\beta,\delta}.
\label{eq:i1_NTLE}
\end{equation}

Substituting  $I_{1}, I_{2}, I_{3}, I_{4}$, into \eqref{eq:entropy_split_NTLE} gives
\[
H(Y)
=
-\log(\beta\lambda)
-j_{\beta,\delta}
-\frac{\beta-1}{\beta}(-\delta)
-\left(K_{\delta}+1-\frac{\delta}{2}\right)
+3\left(1-\frac{\delta}{2}\right).
\]
Simplifying further, we obtain
\[
H(Y)=2-\log(\beta\lambda)-\frac{\delta}{\beta}-j_{\beta,\delta}-K_{\delta},
\]
which is \eqref{eq:shannon_main_NTLE}. Substituting the closed form
\eqref{eq:Kd_closed_NTLE} gives \eqref{eq:shannon_closed_NTLE}.
\end{proof}

\begin{remark}
When $\delta=0$, the distribution reduces to the baseline Logistic-Exponential distribution. Therefore,
\[
K_0=\int_0^1 \log(1)\,dv=0.
\]
As result, the Shannon entropy reduces to
\[
H(Y)=2-\log(\beta\lambda)-j_{\beta,0},
\]
where
\[
j_{\beta,0}=
\int_0^1
\log\!\left[1+\left(\frac{v}{1-v}\right)^{1/\beta}\right]dv .
\]
This expression represents the Shannon entropy of the logistic--exponential distribution with parameters $\lambda$ and $\beta$. The integral $j_{\beta,0}$ does not admit a closed form and may be evaluated numerically for given parameter values.
\end{remark}

\subsection{R\'enyi entropy}

Unlike the classical Shannon entropy, R\'enyi entropy introduces an order parameter that allows different aspects of uncertainty to be emphasized. It has found applications in signal processing, ecology, and machine learning where the degree of uncertainty in complex systems must be quantified \citep{jizba2022}. 

\begin{Proposition}
Let $Y\sim NTLE(\lambda,\beta,\delta)$ with $\lambda>0$, $\beta>0$, and $-1<\delta<1$. For $\rho>0$ and $\rho\neq1$, the R\'enyi entropy of $Y$ is defined by
\[
H_{\rho}(Y)=\frac{1}{1-\rho}\log\left(\int_0^\infty g(y)^\rho\,dy\right),
\]
provided that the integral $\int_0^\infty g(y)^\rho\,dy$ exists.

If $\rho=m$ is an integer with $m\ge2$, then
\begin{align}
\int_{0}^{\infty} g(y)^m\,dy
&=
(\beta\lambda)^{m-1}
\sum_{j=0}^{m-1}\sum_{k=0}^{m}
\binom{m-1}{j}\binom{m}{k}
(1+\delta)^{m-k}(1-\delta)^k
\nonumber\\
&\quad\times
B\!\left(
1+k+\frac{(m-1)(\beta-1)+j}{\beta},
\,
3m-1-k-\frac{(m-1)(\beta-1)+j}{\beta}
\right),
\end{align}
whenever
\[
1+k+\frac{(m-1)(\beta-1)+j}{\beta}>0
\quad\text{and}\quad
3m-1-k-\frac{(m-1)(\beta-1)+j}{\beta}>0
\]
for all $j=0,1,\dots,m-1$ and $k=0,1,\dots,m$.

Consequently,
\begin{align}
H_{m}(Y)
&=
\frac{1}{1-m}
\log\Bigg[
(\beta\lambda)^{m-1}
\sum_{j=0}^{m-1}\sum_{k=0}^{m}
\binom{m-1}{j}\binom{m}{k}
(1+\delta)^{m-k}(1-\delta)^k
\nonumber\\
&\qquad\qquad\qquad\times
B\!\left(
1+k+\frac{(m-1)(\beta-1)+j}{\beta},
\,
3m-1-k-\frac{(m-1)(\beta-1)+j}{\beta}
\right)
\Bigg].
\end{align}
\end{Proposition}

\begin{proof}
By definition, the R\'enyi entropy of order $\rho>0$, $\rho\neq1$, is
\[
H_{\rho}(Y)=\frac{1}{1-\rho}\log\left(\int_{0}^{\infty} g(y)^{\rho}\,dy\right).
\]

To evaluate the quantity $\int_{0}^{\infty} g(y)^{\rho}\,dy$, let
\[
u=(e^{\lambda y}-1)^\beta,\qquad 0<u<\infty.
\]
Then
\[
e^{\lambda y}=1+u^{1/\beta},
\]
and
\[
du=\beta\lambda e^{\lambda y}(e^{\lambda y}-1)^{\beta-1}\,dy
=\beta\lambda(1+u^{1/\beta})u^{(\beta-1)/\beta}\,dy.
\]
Hence,
\[
dy=
\frac{du}{\beta\lambda(1+u^{1/\beta})u^{(\beta-1)/\beta}}.
\]

Under this transformation, the density becomes
\[
g(y)=
\frac{\beta\lambda(1+u^{1/\beta})u^{(\beta-1)/\beta}\left[1+\delta+(1-\delta)u\right]}{(1+u)^3}.
\]
Therefore,
\begin{align*}
g(y)^\rho
&=
(\beta\lambda)^\rho
(1+u^{1/\beta})^\rho
u^{\frac{(\beta-1)\rho}{\beta}}
\left(1+\delta+(1-\delta)u\right)^\rho
(1+u)^{-3\rho},
\end{align*}
and so
\begin{align*}
g(y)^\rho\,dy
&=
(\beta\lambda)^{\rho-1}
u^{\frac{(\beta-1)(\rho-1)}{\beta}}
(1+u^{1/\beta})^{\rho-1}
\left(1+\delta+(1-\delta)u\right)^\rho
(1+u)^{-3\rho}\,du.
\end{align*}
Integrating both sides over $(0,\infty)$ gives
\[
\int_{0}^{\infty} g(y)^{\rho}\,dy
=
(\beta\lambda)^{\rho-1}
\int_{0}^{\infty}
u^{\frac{(\beta-1)(\rho-1)}{\beta}}
(1+u^{1/\beta})^{\rho-1}
\left(1+\delta+(1-\delta)u\right)^{\rho}
(1+u)^{-3\rho}\,du.
\]

Now consider the special case $\rho=m$, where $m\ge2$ is an integer. Since $m-1$ and $m$ are nonnegative integers, we may use the finite binomial expansions
\[
(1+u^{1/\beta})^{m-1}
=
\sum_{j=0}^{m-1}\binom{m-1}{j}u^{j/\beta},
\]
and
\[
\left(1+\delta+(1-\delta)u\right)^m
=
\sum_{k=0}^{m}\binom{m}{k}(1+\delta)^{m-k}(1-\delta)^k u^k.
\]
Substituting these into the previous integral yields
\begin{align*}
\int_{0}^{\infty} g(y)^m\,dy
&=
(\beta\lambda)^{m-1}
\int_{0}^{\infty}
u^{\frac{(\beta-1)(m-1)}{\beta}}
\left(\sum_{j=0}^{m-1}\binom{m-1}{j}u^{j/\beta}\right)
\left(\sum_{k=0}^{m}\binom{m}{k}(1+\delta)^{m-k}(1-\delta)^k u^k\right)
(1+u)^{-3m}\,du
\\
&=
(\beta\lambda)^{m-1}
\sum_{j=0}^{m-1}\sum_{k=0}^{m}
\binom{m-1}{j}\binom{m}{k}
(1+\delta)^{m-k}(1-\delta)^k
\int_{0}^{\infty}
u^{k+\frac{(m-1)(\beta-1)+j}{\beta}}
(1+u)^{-3m}\,du.
\end{align*}

Set
\[
a=
1+k+\frac{(m-1)(\beta-1)+j}{\beta}.
\]
Then
\[
a-1=
k+\frac{(m-1)(\beta-1)+j}{\beta},
\]
and therefore
\[
3m-a=
3m-1-k-\frac{(m-1)(\beta-1)+j}{\beta}.
\]
Hence,
\[
\int_{0}^{\infty}
u^{k+\frac{(m-1)(\beta-1)+j}{\beta}}
(1+u)^{-3m}\,du
=
\int_{0}^{\infty}u^{a-1}(1+u)^{-3m}\,du.
\]
Using the beta integral identity
\[
\int_{0}^{\infty}u^{a-1}(1+u)^{-(a+b)}\,du
=
B(a,b),
\qquad a>0,\quad b>0,
\]
with
\[
b=3m-a
=
3m-1-k-\frac{(m-1)(\beta-1)+j}{\beta},
\]
we obtain
\begin{align*}
\int_{0}^{\infty}
u^{k+\frac{(m-1)(\beta-1)+j}{\beta}}
(1+u)^{-3m}\,du
&=
B\!\left(
1+k+\frac{(m-1)(\beta-1)+j}{\beta},
\,
3m-1-k-\frac{(m-1)(\beta-1)+j}{\beta}
\right),
\end{align*}
provided that
\[
1+k+\frac{(m-1)(\beta-1)+j}{\beta}>0
\quad\text{and}\quad
3m-1-k-\frac{(m-1)(\beta-1)+j}{\beta}>0.
\]
Therefore,
\begin{align*}
\int_{0}^{\infty} g(y)^m\,dy
&=
(\beta\lambda)^{m-1}
\sum_{j=0}^{m-1}\sum_{k=0}^{m}
\binom{m-1}{j}\binom{m}{k}
(1+\delta)^{m-k}(1-\delta)^k
\\
&\quad\times
B\!\left(
1+k+\frac{(m-1)(\beta-1)+j}{\beta},
\,
3m-1-k-\frac{(m-1)(\beta-1)+j}{\beta}
\right).
\end{align*}

Finally, substituting this expression into the definition of the R\'enyi entropy gives
\begin{align*}
H_m(Y)
&=
\frac{1}{1-m}
\log\Bigg[
(\beta\lambda)^{m-1}
\sum_{j=0}^{m-1}\sum_{k=0}^{m}
\binom{m-1}{j}\binom{m}{k}
(1+\delta)^{m-k}(1-\delta)^k
\\
&\qquad\qquad\qquad\times
B\!\left(
1+k+\frac{(m-1)(\beta-1)+j}{\beta},
\,
3m-1-k-\frac{(m-1)(\beta-1)+j}{\beta}
\right)
\Bigg],
\end{align*}
\end{proof}
\subsection{Stochastic ordering}
Stochastic ordering provides a mathematical framework for comparing two random variables in terms of their magnitude, reliability, or variability \citep{shaked2007,belzunce2022}. If one distribution is stochastically larger than another, it tends to produce larger values and may therefore represent a more reliable system.

\begin{Proposition}
Let
$Y_1\sim NTLE(\lambda_1,\beta,\delta_1),\qquad
Y_2\sim NTLE(\lambda_2,\beta,\delta_2),$
where $\beta>0$ is common to both distributions. If $\lambda_1\ge\lambda_2$ and $\delta_1\ge\delta_2$, then $Y_1\le_{st}Y_2.$
\end{Proposition}
\begin{proof}
Using the CDF of NTLE in equation \ref{equation CDF}, we can write 
\begin{equation}
G_r(y)=
\frac{u_r(y)\left[1+\delta_r+u_r(y)\right]}
{\left[1+u_r(y)\right]^2},
\qquad
u_r(y)=\left(e^{\lambda_r y}-1\right)^\beta,
\qquad r=1,2.
\end{equation}
Let
\[
\phi(u,\delta)=\frac{u(1+\delta+u)}{(1+u)^2}, \qquad u>0.
\]
Then, 
\[
G_r(y)=\phi(u_r(y),\delta_r).
\]
Differentiating $\phi$ with respect to $u$ and $\delta$ gives
\[
\frac{\partial \phi(u,\delta)}{\partial u}
=
\frac{1+\delta+(1-\delta)u}{(1+u)^3}>0,
\qquad -1<\delta<1,\; u>0,
\]
and
\[
\frac{\partial \phi(u,\delta)}{\partial \delta}
=
\frac{u}{(1+u)^2}>0.
\]
Hence, $\phi(u,\delta)$ is increasing in both $u$ and $\delta$.

Now, if $\lambda_1\ge \lambda_2$, then for every $y>0$,
\[
u_1(y)=\left(e^{\lambda_1 y}-1\right)^\beta
\ge
\left(e^{\lambda_2 y}-1\right)^\beta
=u_2(y).
\]
Together with $\delta_1\ge \delta_2$, this gives
\[
G_1(y)=\phi(u_1(y),\delta_1)\ge \phi(u_2(y),\delta_2)=G_2(y),
\qquad y>0.
\]
Therefore,
\[
\overline{G}_1(y)\le \overline{G}_2(y), \qquad y>0,
\]
which is equivalent to
\[
Y_1\le_{st} Y_2.
\]
where $\overline{G}_r(y) = 1- G_r(y)$
\end{proof}

\begin{remark}
The above proposition shows that, for fixed $\beta$, the NTLE distribution becomes stochastically smaller as either $\lambda$ or $\delta$ increases.
\end{remark}

\subsection{Mode of the NTLE distribution}
In lifetime analysis, the mode helps identify the most likely failure time of a component or system. It is particularly useful when distributions are skewed, as the mode may provide a more meaningful measure of central tendency than the mean. In practice, the mode can help engineers and reliability analysts determine the most likely time at which failures or events occur.

\begin{theorem}
Let $Y\sim NTLE(\lambda,\beta,\delta)$. Then the mode of the distribution is characterised as follows:

\begin{enumerate}
\item If $0<\beta<1$, then the density is unbounded at the origin and the mode is $y=0$.

\item If $\beta>1$, then the mode is the interior point
\[
y_m=\frac{1}{\lambda}\log(1+x_m),
\]
where $x_m>0$ is a positive root of
\begin{equation}
\frac{1}{1+x}
+\frac{\beta-1}{x}
+\frac{\beta(1-\delta)x^{\beta-1}}{1+\delta+(1-\delta)x^\beta}
-\frac{3\beta x^{\beta-1}}{1+x^\beta}
=0.
\end{equation}

\item If $\beta=1$, then
\[
x_m=-\frac{1+3\delta}{1-\delta}.
\]
Hence, if $\delta<-\tfrac{1}{3}$,
\[
y_m=\frac{1}{\lambda}\log\left(\frac{-4\delta}{1-\delta}\right),
\]
whereas for $\delta\ge -\tfrac{1}{3}$ the mode occurs at $y=0$.
\end{enumerate}
\end{theorem}

\begin{proof}
Let
\[
x=e^{\lambda y}-1, \qquad x>0.
\]
Then
\[
y=\frac{1}{\lambda}\log(1+x),
\]
and the density becomes
\[
g(y(x))
=
\beta\lambda (1+x)x^{\beta-1}
\frac{1+\delta+(1-\delta)x^\beta}{(1+x^\beta)^3}.
\]
Since $x=e^{\lambda y}-1$ is a strictly increasing function of $y$, maximizing $g(y)$ over $y>0$ is equivalent to maximizing $g(y(x))$ over $x>0$.

consider
\[
\ell(x)=\log g(y(x)).
\]
Then
\[
\ell(x)
=
\log(\beta\lambda)+\log(1+x)+(\beta-1)\log x
+\log\left[1+\delta+(1-\delta)x^\beta\right]
-3\log(1+x^\beta).
\]
Differentiating,
\[
\frac{d\ell(x)}{dx}
=
\frac{1}{1+x}
+\frac{\beta-1}{x}
+\frac{\beta(1-\delta)x^{\beta-1}}{1+\delta+(1-\delta)x^\beta}
-\frac{3\beta x^{\beta-1}}{1+x^\beta}.
\]
Therefore, any interior mode must satisfy
\[
\frac{d\ell(x)}{dx}=0,
\]
which yields the stated nonlinear equation.

Next, as $y\downarrow 0$,
\[
e^{\lambda y}-1\sim \lambda y,
\]
so that
\[
g(y)\sim \beta\lambda(1+\delta)(\lambda y)^{\beta-1}.
\]
If $0<\beta<1$, then $g(y)\to\infty$ as $y\downarrow 0$, and hence the mode is at $y=0$. if $\beta>1$, then $g(y)\to 0$ as $y\downarrow 0$, so the mode must be an interior point determined by the modal equation.

Finally, for $\beta=1$, the modal equation reduces to
\[
\frac{1}{1+x}
+\frac{1-\delta}{1+\delta+(1-\delta)x}
-\frac{3}{1+x}
=0,
\]
that is,
\[
-\frac{2}{1+x}
+\frac{1-\delta}{1+\delta+(1-\delta)x}
=0.
\]
Solving for $x$ yields
\[
x_m=-\frac{1+3\delta}{1-\delta}.
\]
Thus $x_m>0$ if and only if $\delta<-\tfrac{1}{3}$, and in that case
\[
y_m=\frac{1}{\lambda}\log(1+x_m)
=
\frac{1}{\lambda}\log\left(\frac{-4\delta}{1-\delta}\right).
\]
When $\delta\ge -\tfrac{1}{3}$, there is no positive interior solution, so the mode occurs at the boundary point $y=0$.
\end{proof}

\subsection{Stress-strength reliability}
The stress-strength reliability parameter measures the probability that a system's strength exceeds the applied stress, usually expressed as $R=p(Y>X)$ where $Y$ represents strength and $X$ represents stress. It is used to determine the probability that a component will perform satisfactorily under operational conditions \citep{kotz2003, hassan2023reliability}. 
The stress--strength reliability for two independent random variables $X$ and $Y$ is defined as:
\begin{equation*}
R=p(Y<X)=\int_{0}^{\infty}G_Y(x)\,g_X(x)\,dx.
\end{equation*}
\begin{Proposition}
Let 
$X\sim NTLE(\lambda_1,\beta_1,\delta_1),
\quad
Y\sim NTLE(\lambda_2,\beta_2,\delta_2),
$ be independent random variables, where $X$ denotes strength and $Y$ denotes stress. Then, the stress--strength reliability if
$
\lambda_1=\lambda_2=\lambda \; \text{and}  \;\beta_1=\beta_2=\beta,
$
is
\begin{equation}
R=\frac{1}{2}+\frac{\delta_2-\delta_1}{6}.
\end{equation}
\end{Proposition}

\begin{proof}
By independence,
\[
R=p(Y<X)=\int_{0}^{\infty}G_Y(x)\,g_X(x)\,dx.
\]
For the NTLE model, this becomes
\begin{align}
R
&=
\int_{0}^{\infty}
\frac{(e^{\lambda_2 x}-1)^{\beta_2}\left[1+\delta_2+(e^{\lambda_2 x}-1)^{\beta_2}\right]}
{\left[1+(e^{\lambda_2 x}-1)^{\beta_2}\right]^2}
\nonumber\\
&\qquad\qquad\times
\frac{\beta_1\lambda_1 e^{\lambda_1 x}(e^{\lambda_1 x}-1)^{\beta_1-1}
\left[(1+(e^{\lambda_1 x}-1)^{\beta_1})+\delta_1(1-(e^{\lambda_1 x}-1)^{\beta_1})\right]}
{\left(1+(e^{\lambda_1 x}-1)^{\beta_1}\right)^3}
\,dx.
\end{align}
This gives the general form.

Now assume $\lambda_1=\lambda_2=\lambda$ and $\beta_1=\beta_2=\beta$. Let
\[
u=(e^{\lambda x}-1)^\beta.
\]
Then
\[
du=\beta\lambda e^{\lambda x}(e^{\lambda x}-1)^{\beta-1}\,dx,
\]
and
\[
G_Y(x)=\frac{u(1+\delta_2+u)}{(1+u)^2},
\qquad
g_X(x)\,dx=
\frac{1+\delta_1+(1-\delta_1)u}{(1+u)^3}\,du.
\]
Hence,
\[
R
=
\int_{0}^{\infty}
\frac{u(1+\delta_2+u)\left[1+\delta_1+(1-\delta_1)u\right]}
{(1+u)^5}\,du.
\]
Simplifying the numerator
\[
u(1+\delta_2+u)\left[1+\delta_1+(1-\delta_1)u\right]
=
(1+\delta_2)(1+\delta_1)u
+\left[(1+\delta_2)(1-\delta_1)+(1+\delta_1)\right]u^2
+(1-\delta_1)u^3.
\]
Therefore,
\begin{align}
R
&=
(1+\delta_2)(1+\delta_1)\int_{0}^{\infty}\frac{u}{(1+u)^5}\,du
+\left[(1+\delta_2)(1-\delta_1)+(1+\delta_1)\right]
\int_{0}^{\infty}\frac{u^2}{(1+u)^5}\,du
\nonumber\\
&\qquad
+(1-\delta_1)\int_{0}^{\infty}\frac{u^3}{(1+u)^5}\,du.
\end{align}

Using
\[
\int_{0}^{\infty}\frac{u}{(1+u)^5}\,du=B(2,3)=\frac{1}{12},
\qquad
\int_{0}^{\infty}\frac{u^2}{(1+u)^5}\,du=B(3,2)=\frac{1}{12},
\]
and
\[
\int_{0}^{\infty}\frac{u^3}{(1+u)^5}\,du=B(4,1)=\frac{1}{4},
\]
we obtain
\begin{align*}
R
&=
\frac{(1+\delta_2)(1+\delta_1)}{12}
+\frac{(1+\delta_2)(1-\delta_1)+(1+\delta_1)}{12}
+\frac{1-\delta_1}{4}.
\end{align*}
Simplifying, we obtain
\[
R=\frac{1}{2}+\frac{\delta_2-\delta_1}{6}.
\]
This completes the proof.
\end{proof}

\begin{remark}
When $\delta_1=\delta_2$, the above formula reduces to
\[
R=\frac{1}{2},
\]
which is consistent with symmetry when the strength and stress variables belong to the same NTLE model.
\end{remark}

\subsection{Residual life functions}
Residual life functions are widely used in reliability and survival analysis as they describe the expected remaining lifetime of a component that has survived up to a given time \citep{lai2022,navarro2021aging}. 
The mean residual life (MRL) function is
\begin{equation}
m(t)=e(Y-t\,|\,Y>t)
=\frac{1}{\overline{G}(t)}
\int_{t}^{\infty}\overline{G}(y)\,dy .
\end{equation}
while  the reversed residual life (RRL) function is defined as
\begin{equation}
r(t)=e(t-Y\,|\,Y\le t)
=\frac{1}{G(t)}\int_{0}^{t}G(y)\,dy .
\end{equation}

\begin{Proposition}[Mean residual life function]
Let $Y\sim NTLE(\lambda,\beta,\delta)$ the mean residual life is
\begin{equation}
m(t)=
\frac{1}{\overline{G}(t)}
\int_{t}^{\infty}
\frac{1+(1-\delta)u(y)}{(1+u(y))^2}\,dy .
\end{equation}
\end{Proposition}

\begin{proof}
Let 
\[
u=(e^{\lambda y}-1)^\beta .
\]
Then
\[
du=\beta\lambda e^{\lambda y}(e^{\lambda y}-1)^{\beta-1}dy .
\]

Since
\[
e^{\lambda y}=1+u^{1/\beta},
\]
we obtain
\[
dy=
\frac{du}{\beta\lambda (1+u^{1/\beta})u^{(\beta-1)/\beta}} .
\]

Hence
\begin{align}
\int_{t}^{\infty}\overline{G}(y)dy
&=
\int_{u_t}^{\infty}
\frac{1+(1-\delta)u}{(1+u)^2}
\frac{du}{\beta\lambda (1+u^{1/\beta})u^{(\beta-1)/\beta}},
\end{align}
where the survival function is
\begin{equation}
\overline{G}(y)=1-G(y)
=\frac{1+(1-\delta)u(y)}{(1+u(y))^2}.
\end{equation}

Therefore
\begin{equation}
m(t)=
\frac{1}{\overline{G}(t)}
\frac{1}{\beta\lambda}
\int_{u_t}^{\infty}
\frac{1+(1-\delta)u}
{(1+u)^2(1+u^{1/\beta})u^{(\beta-1)/\beta}}du .
\end{equation}
\end{proof}

\begin{Proposition}[Reversed residual life function]
The reversed residual life (RRL) function NTLE distribution
is given as \begin{equation}
r(t)=
\frac{1}{G(t)}
\int_{0}^{t}
\frac{u(y)[1+\delta+u(y)]}{(1+u(y))^2}dy .
\end{equation}
\end{Proposition}

\begin{proof}
Using the transformation
\[
u=(e^{\lambda y}-1)^\beta ,
\]
we have
\[
dy=\frac{du}{\beta\lambda (1+u^{1/\beta})u^{(\beta-1)/\beta}} .
\]

Therefore
\begin{align}
\int_{0}^{t}G(y)dy
&=
\frac{1}{\beta\lambda}
\int_{0}^{u_t}
\frac{u(1+\delta+u)}
{(1+u)^2(1+u^{1/\beta})u^{(\beta-1)/\beta}}du ,
\end{align}
\end{proof}

\subsection{Incomplete moments}
Incomplete moments describe the expected value of a random variable over a restricted portion of its support. In reliability applications, incomplete moments help quantify the contribution of observations below a certain threshold, which is important when studying early failures or truncated data. They are also widely used in economic inequality measures and insurance risk modelling \citep{elgarhy2022}.
The $r$th incomplete moment is given by:
\begin{equation}
\mu_r(t)=e(Y^r\mathbf{1}_{\{Y\le t\}})
=\int_{0}^{t} y^r g(y)\,dy .
\end{equation}
\begin{Proposition}
Let $Y\sim NTLE(\lambda,\beta,\delta)$. The $r$th incomplete moment is given as

\begin{align}
\mu_r(t)
&=
\int_{0}^{u_t}
\left(\frac{1}{\lambda}\log(1+u^{1/\beta})\right)^r
\frac{1+\delta+(1-\delta)u}{(1+u)^3}du ,
\end{align}
\end{Proposition}

\begin{proof}
By definition;
\begin{equation}
\mu_r(t)=e(Y^r\mathbf{1}_{\{Y\le t\}})
=\int_{0}^{t} y^r g(y)\,dy .
\end{equation}
Let $u=(e^{\lambda y}-1)^\beta$

Then, 
$y=\frac{1}{\lambda}\log\!\left(1+u^{1/\beta}\right)$ \qquad and $du=\beta\lambda e^{\lambda y}(e^{\lambda y}-1)^{\beta-1}\,dy$

\noindent Using the PDF in equation \ref{equation PDF} and noting that the term $ \beta\lambda e^{\lambda y}(e^{\lambda y}-1)^{\beta-1}dy$ also appears in the density. 
Hence, the Jacobian cancels the corresponding component of the density, and we obtain 
\begin{align*}
\mu_r(t)
&=
\int
\left(\frac{1}{\lambda}\log(1+u^{1/\beta})\right)^r
\frac{(1+u)+\delta(1-u)}{(1+u)^3}
\,du .
\end{align*}

\noindent Note that
\[
(1+u)+\delta(1-u)=1+\delta+(1-\delta)u,
\]
The integral becomes
\begin{equation}
\mu_r(t)
=
\int_{0}^{u_t}
\left(\frac{1}{\lambda}\log(1+u^{1/\beta})\right)^r
\frac{1+\delta+(1-\delta)u}{(1+u)^3}\,du ,
\end{equation}
\end{proof}

\subsection{Bonferroni and Lorenz curves} 
Bonferroni and Lorenz curves are used tools for measuring inequality and concentration within a distribution. The Lorenz curve describes the cumulative proportion of the total variable accounted for by the bottom fraction of the population, while the Bonferroni curve provides a related measure of concentration \citep{arshad}.

\begin{Proposition}
Let $F^{-1}(p)$ denote the quantile function of the NTLE distribution.  
The Bonferroni curve is
\begin{equation}
B(p)=\frac{1}{p\mu}\int_{0}^{F^{-1}(p)}y\,g(y)\,dy ,
\end{equation}
while the Lorenz curve is
\begin{equation}
L(p)=\frac{1}{\mu}\int_{0}^{F^{-1}(p)}y\,g(y)\,dy ,
\end{equation}
where $\mu=E(Y)$.

Using the incomplete first moment $\mu_1(t)$,
\begin{equation}
L(p)=\frac{\mu_1(F^{-1}(p))}{\mu},
\qquad
B(p)=\frac{\mu_1(F^{-1}(p))}{p\mu}.
\end{equation}
\end{Proposition}

For the NTLE distribution,
\[
\mu_1(F^{-1}(p))
=
\frac{1}{\lambda}
\int_0^{u_p}
\log(1+u^{1/\beta})
\frac{1+\delta+(1-\delta)u}{(1+u)^3}\,du,
\]
where
\[
u_p=(e^{\lambda F^{-1}(p)}-1)^\beta.
\]

\begin{remark}
The Bonferroni and Lorenz curves provide useful measures of concentration and inequality for the NTLE distribution and can be evaluated once the incomplete first moment is obtained.
\end{remark}

\section{Methods of parameter estimation for NTLE}

Let $Y_1, Y_2, \ldots, Y_n$ be a random sample drawn from a continuous distribution with probability density function $g(y;\theta)$ and cumulative distribution function $G(xy;\theta)$, where $\theta$ denotes an unknown parameter or a vector of parameters to be estimated. Various statistical techniques have been proposed for estimating such parameters. In this study, ten estimation methods are considered due to their widespread application in statistical inference and distributional modeling. The parameters $(\lambda,\beta,\delta)$ are estimated using several estimation techniques described in sections 3.1 - 3.10.

\subsection{Maximum Likelihood estimation}
Maximum likelihood estimation determines parameter values that maximize the likelihood of observing the data \citep{Fisher1922, casellaBerger2002}

The likelihood function for the NTLE distribution based on the observed sample is

\begin{equation}
L(\lambda,\beta,\delta)=\prod_{i=1}^{n} g(y_i;\lambda,\beta,\delta).
\end{equation}

Substituting the density function yields

\begin{equation}
L(\lambda,\beta,\delta)=
\prod_{i=1}^{n}
\frac{\beta\lambda e^{\lambda y_i}(e^{\lambda y_i}-1)^{\beta-1}
[(1+(e^{\lambda y_i}-1)^\beta)+\delta(1-(e^{\lambda y_i}-1)^\beta)]}
{(1+(e^{\lambda y_i}-1)^\beta)^3}.
\end{equation}

Taking logarithms, the log-likelihood function becomes
\begin{align*}
\ell(\lambda,\beta,\delta)
&=n\ln\beta+n\ln\lambda+\lambda\sum_{i=1}^{n}y_i
+(\beta-1)\sum_{i=1}^{n}\ln(e^{\lambda y_i}-1) \\
&\quad +\sum_{i=1}^{n}\ln\Big[(1+(e^{\lambda y_i}-1)^\beta)
+\delta(1-(e^{\lambda y_i}-1)^\beta)\Big] \\
&\quad -3\sum_{i=1}^{n}\ln\Big[1+(e^{\lambda y_i}-1)^\beta\Big].
\end{align*}

The score vector is obtained by differentiating the log-likelihood with respect to the parameters \citep{adesegun2023tle}:

\[
U(\boldsymbol{\theta})=
\left(
\frac{\partial\ell}{\partial\lambda},
\frac{\partial\ell}{\partial\beta},
\frac{\partial\ell}{\partial\delta}
\right)^T.
\]

The likelihood equations are obtained by solving

\begin{equation}
\frac{\partial\ell}{\partial\lambda}=0, \quad
\frac{\partial\ell}{\partial\beta}=0, \quad
\frac{\partial\ell}{\partial\delta}=0.
\end{equation}

These equations do not admit closed-form solutions, numerical procedures such as the Newton--Raphson algorithm or quasi-Newton optimization methods are employed \citep{adesegun2023tle}.

\subsubsection{Observed Fisher Information Matrix}

Let

\[
i(\boldsymbol{\theta})=
-\frac{\partial^2\ell}{\partial\boldsymbol{\theta}\partial\boldsymbol{\theta}^T}
\]

denote the observed Fisher information matrix. explicitly,

\[
i(\boldsymbol{\theta})=
\begin{pmatrix}
-\frac{\partial^2\ell}{\partial\lambda^2} &
-\frac{\partial^2\ell}{\partial\lambda\partial\beta} &
-\frac{\partial^2\ell}{\partial\lambda\partial\delta} \\

-\frac{\partial^2\ell}{\partial\beta\partial\lambda} &
-\frac{\partial^2\ell}{\partial\beta^2} &
-\frac{\partial^2\ell}{\partial\beta\partial\delta} \\

-\frac{\partial^2\ell}{\partial\delta\partial\lambda} &
-\frac{\partial^2\ell}{\partial\delta\partial\beta} &
-\frac{\partial^2\ell}{\partial\delta^2}
\end{pmatrix}.
\]

Under asymptotic normal conditions, the maximum likelihood estimator satisfies

\[
\sqrt{n}(\hat{\boldsymbol{\theta}}-\boldsymbol{\theta})
\rightarrow N(0,i^{-1}(\boldsymbol{\theta})).
\]

Consequently, approximate confidence intervals for the parameters may be constructed using

\[
\hat{\theta}_i \pm z_{\alpha/2}\sqrt{\widehat{Var}(\hat{\theta}_i)}.
\]

\subsection{Method of Moments estimation}
The method of moments (MME) provides an alternative approach for estimating
unknown parameters by equating the sample moments with their corresponding
population moments \citep{pearson1894}. 

The $k$th sample moment about the origin is defined as
\begin{equation}
m_k=\frac{1}{n}\sum_{i=1}^{n}Y_i^k, \qquad k=1,2,3.
\end{equation}

Let
\[
\mu_k=e(Y^k)
\]
denote the $k$th theoretical moment of the NTLE distribution. Then, 
\begin{equation}
\mu_k = e(Y^k)=\int_{0}^{\infty} y^k g(y)\,dy ,
\end{equation}
Using the transformation
\[
u=(e^{\lambda y}-1)^\beta .
\]
Then
\[
y=\frac{1}{\lambda}\log\left(1+u^{1/\beta}\right),
\]
and
\[
du=\beta\lambda e^{\lambda y}(e^{\lambda y}-1)^{\beta-1}dy .
\]
Hence,
\begin{equation}
\mu_k=
\int_{0}^{\infty}
\left(\frac{1}{\lambda}\log(1+u^{1/\beta})\right)^k
\frac{(1+u)+\delta(1-u)}{(1+u)^3}
\,du .
\end{equation}

Therefore, the first three population moments of the NTLE distribution are

\begin{align*}
\mu_1 &=
\int_{0}^{\infty}
\frac{1}{\lambda}\log(1+u^{1/\beta})
\frac{(1+u)+\delta(1-u)}{(1+u)^3}\,du, \\
\\
\mu_2 &=
\int_{0}^{\infty}
\left(\frac{1}{\lambda}\log(1+u^{1/\beta})\right)^2
\frac{(1+u)+\delta(1-u)}{(1+u)^3}\,du, \\
\\
\mu_3 &=
\int_{0}^{\infty}
\left(\frac{1}{\lambda}\log(1+u^{1/\beta})\right)^3
\frac{(1+u)+\delta(1-u)}{(1+u)^3}\,du .
\end{align*}

The method of moments estimators of the parameters
$(\lambda,\beta,\delta)$ are obtained by solving the equations;

\begin{align}
\frac{1}{n}\sum_{i=1}^{n}Y_i
&=
\int_{0}^{\infty}
\frac{1}{\lambda}\log(1+u^{1/\beta})
\frac{(1+u)+\delta(1-u)}{(1+u)^3}\,du
\label{eqt 1},
\\
\frac{1}{n}\sum_{i=1}^{n}Y_i^2
&=
\int_{0}^{\infty}
\left(\frac{1}{\lambda}\log(1+u^{1/\beta})\right)^2
\frac{(1+u)+\delta(1-u)}{(1+u)^3}\,du,
\label{eqt 11}
\\
\frac{1}{n}\sum_{i=1}^{n}Y_i^3
&=
\int_{0}^{\infty}
\left(\frac{1}{\lambda}\log(1+u^{1/\beta})\right)^3
\frac{(1+u)+\delta(1-u)}{(1+u)^3}\,du .
\label{eqt 111}
\end{align}

Since these equations \ref{eqt 1}, \ref{eqt 11} and  \ref{eqt 111} do not give a close form. Hence, the parameter estimates $(\hat{\lambda},\hat{\beta},\hat{\delta})$
are obtained numerically.

\subsection{Least Squares Estimation (LSE)}

Least squares estimation for distribution parameters is based on minimizing the squared difference between the theoretical distribution function and the empirical distribution function \citep{Swain1988} i.e. 

\begin{equation}
\sum_{i=1}^{n}
\left[
G(y_{(i)};\lambda,\beta,\delta) - \frac{i}{n+1}
\right]^2,
\end{equation}

where $G(y_{(i)};\lambda,\beta,\delta)$ is the NTLE cumulative distribution function.

The least squares objective function is therefore given by
\begin{equation*}
S(\lambda,\beta,\delta)
=
\sum_{i=1}^{n}
\left[
\frac{(e^{\lambda y_{(i)}}-1)^\beta\left\{1+\delta+(e^{\lambda y_{(i)}}-1)^\beta\right\}}
{\left\{1+(e^{\lambda y_{(i)}}-1)^\beta\right\}^2}
-\frac{i}{n+1}
\right]^2.
\end{equation*}

Let
\begin{equation*}
u_i=(e^{\lambda y_{(i)}}-1)^\beta, \qquad i=1,2,\ldots,n.
\end{equation*}
Then
\begin{equation*}
G(y_{(i)};\lambda,\beta,\delta)=\frac{u_i(1+\delta+u_i)}{(1+u_i)^2},
\end{equation*}
and the objective function becomes
\begin{equation*}
S(\lambda,\beta,\delta)
=
\sum_{i=1}^{n}
\left[
\frac{u_i(1+\delta+u_i)}{(1+u_i)^2}
-\frac{i}{n+1}
\right]^2.
\end{equation*}

Writing
\begin{equation*}
r_i=
\frac{u_i(1+\delta+u_i)}{(1+u_i)^2}
-\frac{i}{n+1},
\end{equation*}
we have
\begin{equation*}
S(\lambda,\beta,\delta)=\sum_{i=1}^{n} r_i^2.
\end{equation*}
Hence,
\begin{equation*}
\frac{\partial S}{\partial \theta}
=
2\sum_{i=1}^{n} r_i \frac{\partial r_i}{\partial \theta},
\qquad
\theta\in\{\lambda,\beta,\delta\}.
\end{equation*}

Now,
\begin{equation*}
\frac{\partial G(y_{(i)};\lambda,\beta,\delta)}{\partial u_i}
=
\frac{1+\delta+(1-\delta)u_i}{(1+u_i)^3},
\end{equation*}
so that
\begin{equation*}
\frac{\partial r_i}{\partial \lambda}
=
\frac{1+\delta+(1-\delta)u_i}{(1+u_i)^3}
\frac{\partial u_i}{\partial \lambda},
\end{equation*}
\begin{equation*}
\frac{\partial r_i}{\partial \beta}
=
\frac{1+\delta+(1-\delta)u_i}{(1+u_i)^3}
\frac{\partial u_i}{\partial \beta},
\end{equation*}
and
\begin{equation*}
\frac{\partial r_i}{\partial \delta}
=
\frac{u_i}{(1+u_i)^2}.
\end{equation*}

Since
\begin{equation*}
u_i=(e^{\lambda y_{(i)}}-1)^\beta,
\end{equation*}
It follows that
\begin{equation*}
\frac{\partial u_i}{\partial \lambda}
=
\beta y_{(i)} e^{\lambda y_{(i)}}(e^{\lambda y_{(i)}}-1)^{\beta-1},
\end{equation*}
and
\begin{equation*}
\frac{\partial u_i}{\partial \beta}
=
(e^{\lambda y_{(i)}}-1)^\beta \log(e^{\lambda y_{(i)}}-1)
=
u_i\log(e^{\lambda y_{(i)}}-1).
\end{equation*}

Therefore, the normal equations for the least squares estimators are
\begin{equation}
\sum_{i=1}^{n}
r_i
\frac{1+\delta+(1-\delta)u_i}{(1+u_i)^3}
\beta y_{(i)} e^{\lambda y_{(i)}}(e^{\lambda y_{(i)}}-1)^{\beta-1}
=0,
\end{equation}
\begin{equation}
\sum_{i=1}^{n}
r_i
\frac{1+\delta+(1-\delta)u_i}{(1+u_i)^3}
u_i\log(e^{\lambda y_{(i)}}-1)
=0,
\end{equation}
and
\begin{equation}
\sum_{i=1}^{n}
r_i
\frac{u_i}{(1+u_i)^2}
=0.
\end{equation}

These equations are nonlinear in $\lambda$, $\beta$ and $\delta$,
closed-form solutions are not available, estimates are obtained numerically using iterative optimisation methods.

\subsection{Weighted Least Squares Estimation}

The weighted least squares (WLSE) method improves the ordinary least squares estimator by assigning different weights to the squared deviations between the theoretical and empirical distribution functions \citep{Swain1988}.. This approach accounts for the fact that the variance of the empirical distribution function depends on the order of the observations.

The WLSE minimizes

\begin{equation}
\sum_{i=1}^{n}
w_i
\left[
G(y_{(i)};\lambda,\beta,\delta) - \frac{i}{n+1}
\right]^2,
\end{equation}

where $w_i$ are weights chosen based on the variance of the empirical distribution function.

A commonly used weight function is

\begin{equation*}
w_i=\frac{(n+1)^2(n+2)}{i(n-i+1)}, \qquad i=1,2,\ldots,n.
\end{equation*}

The weighted least squares objective function is therefore defined as

\begin{equation*}
W(\lambda,\beta,\delta)
=
\sum_{i=1}^{n}
w_i
\left[
\frac{(e^{\lambda y_{(i)}}-1)^\beta\left\{1+\delta+(e^{\lambda y_{(i)}}-1)^\beta\right\}}
{\left\{1+(e^{\lambda y_{(i)}}-1)^\beta\right\}^2}
-\frac{i}{n+1}
\right]^2 .
\end{equation*}

For convenience, let

\begin{equation*}
u_i=(e^{\lambda y_{(i)}}-1)^\beta, \qquad i=1,2,\ldots,n.
\end{equation*}

Then

\begin{equation*}
G(y_{(i)};\lambda,\beta,\delta)=\frac{u_i(1+\delta+u_i)}{(1+u_i)^2},
\end{equation*}

and the weighted least squares objective function becomes

\begin{equation*}
W(\lambda,\beta,\delta)
=
\sum_{i=1}^{n}
w_i
\left[
\frac{u_i(1+\delta+u_i)}{(1+u_i)^2}
-\frac{i}{n+1}
\right]^2 .
\end{equation*}

Let

\begin{equation*}
r_i=
\frac{u_i(1+\delta+u_i)}{(1+u_i)^2}
-\frac{i}{n+1}.
\end{equation*}

Then

\begin{equation*}
W(\lambda,\beta,\delta)=\sum_{i=1}^{n} w_i r_i^2.
\end{equation*}

The weighted least squares estimators
$\hat{\lambda}_{WLSE}$, $\hat{\beta}_{WLSE}$ and $\hat{\delta}_{WLSE}$
are obtained by minimizing $W(\lambda,\beta,\delta)$ with respect to
$\lambda$, $\beta$ and $\delta$. Thus,

\begin{equation*}
\frac{\partial W}{\partial \theta}
=
2\sum_{i=1}^{n} w_i r_i \frac{\partial r_i}{\partial \theta},
\qquad
\theta\in\{\lambda,\beta,\delta\}.
\end{equation*}

Using

\begin{equation*}
\frac{\partial G(y_{(i)};\lambda,\beta,\delta)}{\partial u_i}
=
\frac{1+\delta+(1-\delta)u_i}{(1+u_i)^3},
\end{equation*}

we obtain

\begin{equation*}
\frac{\partial r_i}{\partial \lambda}
=
\frac{1+\delta+(1-\delta)u_i}{(1+u_i)^3}
\frac{\partial u_i}{\partial \lambda},
\end{equation*}

\begin{equation*}
\frac{\partial r_i}{\partial \beta}
=
\frac{1+\delta+(1-\delta)u_i}{(1+u_i)^3}
\frac{\partial u_i}{\partial \beta},
\end{equation*}

and

\begin{equation*}
\frac{\partial r_i}{\partial \delta}
=
\frac{u_i}{(1+u_i)^2}.
\end{equation*}

Since

\begin{equation*}
u_i=(e^{\lambda y_{(i)}}-1)^\beta,
\end{equation*}

It follows that

\begin{equation*}
\frac{\partial u_i}{\partial \lambda}
=
\beta y_{(i)} e^{\lambda y_{(i)}}(e^{\lambda y_{(i)}}-1)^{\beta-1},
\end{equation*}

and

\begin{equation*}
\frac{\partial u_i}{\partial \beta}
=
u_i\log(e^{\lambda y_{(i)}}-1).
\end{equation*}

Therefore, the estimating equations for the weighted least squares estimators are

\begin{equation*}
\sum_{i=1}^{n}
w_i r_i
\frac{1+\delta+(1-\delta)u_i}{(1+u_i)^3}
\beta y_{(i)} e^{\lambda y_{(i)}}(e^{\lambda y_{(i)}}-1)^{\beta-1}
=0,
\end{equation*}

\begin{equation*}
\sum_{i=1}^{n}
w_i r_i
\frac{1+\delta+(1-\delta)u_i}{(1+u_i)^3}
u_i\log(e^{\lambda y_{(i)}}-1)
=0,
\end{equation*}

and

\begin{equation*}
\sum_{i=1}^{n}
w_i r_i
\frac{u_i}{(1+u_i)^2}
=0.
\end{equation*}

The weighted least squares estimates of the NTLE parameters are obtained numerically using iterative optimization algorithms.

\subsection{Maximum Product of Spacings Estimation (MPS)}

In this method, parameters are estimated by maximizing the geometric mean of spacings between successive values of the distribution function. \citep{cheng1979maximum, kurdi2023bayesian}. This spacing is defined as:

\begin{equation}
D_i=G(y_{(i)})-G(y_{(i-1)}), \quad i=1,2,\ldots,n+1,
\end{equation}

where

\[
G(y_{(0)})=0, \qquad G(y_{(n+1)})=1.
\]

The product of spacings is

\begin{equation*}
p(\lambda,\beta,\delta)=\prod_{i=1}^{n+1} D_i.
\end{equation*}

Instead of maximizing $p$, it is convenient to maximize its logarithm

\begin{equation*}
M(\lambda,\beta,\delta)=\sum_{i=1}^{n+1}\log D_i.
\end{equation*}

Substituting the NTLE distribution function yields

\begin{equation*}
D_i =
\frac{(e^{\lambda y_{(i)}}-1)^\beta
\left[1+\delta+(e^{\lambda y_{(i)}}-1)^\beta\right]}
{\left[1+(e^{\lambda y_{(i)}}-1)^\beta\right]^2}
-
\frac{(e^{\lambda y_{(i-1)}}-1)^\beta
\left[1+\delta+(e^{\lambda y_{(i-1)}}-1)^\beta\right]}
{\left[1+(e^{\lambda y_{(i-1)}}-1)^\beta\right]^2}.
\end{equation*}

The MPS estimators $(\hat{\lambda}_{MPS},\hat{\beta}_{MPS},\hat{\delta}_{MPS})$ are obtained by solving

\begin{equation}
\frac{\partial M}{\partial \lambda}=0,\quad
\frac{\partial M}{\partial \beta}=0,\quad
\frac{\partial M}{\partial \delta}=0.
\end{equation}

These equations are nonlinear and therefore solved numerically.

\subsection{Anderson-Darling Estimation}
The Anderson-Darling estimation method is based on minimizing the Anderson-Darling goodness-of-fit statistic between the empirical distribution function and the theoretical distribution function \citep{andersonDarling1952}

The Anderson-Darling objective function is

\begin{equation}
a(\lambda,\beta,\delta)
=
-n-\frac{1}{n}
\sum_{i=1}^{n}
(2i-1)
\left[
\log G(y_{(i)})
+
\log\left(1-G(y_{(n+1-i)})\right)
\right].
\end{equation}

The estimates
$(\hat{\lambda}_{ADE},\hat{\beta}_{ADE},\hat{\delta}_{ADE})$
are obtained by minimizing $a(\lambda,\beta,\delta)$ with respect to the model parameters.

Substituting the NTLE cumulative distribution function gives

\begin{align*}
a(\lambda,\beta,\delta)
&=
-n-\frac{1}{n}
\sum_{i=1}^{n}
(2i-1)
\Bigg[
\log\left(
\frac{(e^{\lambda y_{(i)}}-1)^\beta
\left[1+\delta+(e^{\lambda y_{(i)}}-1)^\beta\right]}
{\left[1+(e^{\lambda y_{(i)}}-1)^\beta\right]^2}
\right)
\\
&\qquad+
\log\left(
1-
\frac{(e^{\lambda y_{(n+1-i)}}-1)^\beta
\left[1+\delta+(e^{\lambda y_{(n+1-i)}}-1)^\beta\right]}
{\left[1+(e^{\lambda y_{(n+1-i)}}-1)^\beta\right]^2}
\right)
\Bigg].
\end{align*}

Again, let

\begin{equation*}
u_i=(e^{\lambda y_{(i)}}-1)^\beta, \qquad i=1,2,\ldots,n.
\end{equation*}

Then

\begin{equation*}
G(y_{(i)})=\frac{u_i(1+\delta+u_i)}{(1+u_i)^2},
\end{equation*}

and

\begin{equation*}
1-G(y_{(i)})=\frac{1+(1-\delta)u_i}{(1+u_i)^2}.
\end{equation*}

Thus the Anderson-Darling objective function becomes

\begin{equation*}
a(\lambda,\beta,\delta)
=
-n-\frac{1}{n}
\sum_{i=1}^{n}
(2i-1)
\left[
\log\left(\frac{u_i(1+\delta+u_i)}{(1+u_i)^2}\right)
+
\log\left(\frac{1+(1-\delta)u_{n+1-i}}{(1+u_{n+1-i})^2}\right)
\right].
\end{equation*}

The Anderson-Darling estimators
$\hat{\lambda}_{ADE}$, $\hat{\beta}_{ADE}$ and $\hat{\delta}_{ADE}$
are obtained by minimizing $a(\lambda,\beta,\delta)$ with respect to
$\lambda$, $\beta$ and $\delta$. Hence, the estimators satisfy

\begin{equation*}
\frac{\partial a}{\partial \lambda}=0,
\qquad
\frac{\partial a}{\partial \beta}=0,
\qquad
\frac{\partial a}{\partial \delta}=0.
\end{equation*}

Since

\begin{equation*}
u_i=(e^{\lambda y_{(i)}}-1)^\beta,
\end{equation*}

the required derivatives involve

\begin{equation*}
\frac{\partial u_i}{\partial \lambda}
=
\beta y_{(i)} e^{\lambda y_{(i)}}(e^{\lambda y_{(i)}}-1)^{\beta-1},
\end{equation*}

and

\begin{equation}
\frac{\partial u_i}{\partial \beta}
=
u_i \log(e^{\lambda y_{(i)}}-1).
\end{equation}

Also, the Anderson-Darling estimators of the NTLE parameters are obtained numerically using iterative optimization techniques.

\subsection{Cramer-von Mises estimation (CVME)}

The Cramer-von Mises estimation method is another goodness-of-fit based approach that measures the squared difference between the empirical distribution function and the theoretical distribution function \citep{cramer1928}.

For ordered observations $y_{(1)}, y_{(2)}, \ldots, y_{(n)}$, the Cramer-von Mises objective function is defined as

\begin{equation}
W^{2}(\lambda,\beta,\delta)
=
\frac{1}{12n}
+
\sum_{i=1}^{n}
\left[
G(y_{(i)};\lambda,\beta,\delta)
-
\frac{2i-1}{2n}
\right]^2,
\end{equation}

where $G(y;\lambda,\beta,\delta)$ denotes the cumulative distribution function of the NTLE distribution. Compared with Anderson-Darling estimation, the Cramér-von Mises method distributes weight more evenly across the entire distribution, rather than emphasizing the tails.

Substituting the NTLE cumulative distribution function gives

\begin{equation}
c(\lambda,\beta,\delta)
=
\frac{1}{12n}
+
\sum_{i=1}^{n}
\left[
\frac{(e^{\lambda y_{(i)}}-1)^\beta
\left\{1+\delta+(e^{\lambda y_{(i)}}-1)^\beta\right\}}
{\left\{1+(e^{\lambda y_{(i)}}-1)^\beta\right\}^2}
-
\frac{2i-1}{2n}
\right]^2.
\end{equation}

The cramér--von Mises estimators $(\hat{\lambda}_{CVME},\hat{\beta}_{CVME},\hat{\delta}_{CVME})$ are obtained by minimizing $c(\lambda,\beta,\delta)$ with respect to the parameters. Since the resulting equations are nonlinear, numerical optimization methods are required.

\subsection{Percentile Estimation (PCE)}

The percentile estimation method is based on matching theoretical percentiles of the distribution with sample percentiles obtained from the ordered observations \citep{Kao1958}. it is particularly useful in situations where the distribution's quantile structure is of primary interest.

Let $y_{(i)}$ denote the $i^{\text{th}}$ ordered observation and let the corresponding theoretical cumulative probability be approximated by

\begin{equation}
p_i = \frac{i}{n+1}.
\end{equation}

The theoretical quantiles satisfy

\begin{equation*}
G(y_{(i)};\lambda,\beta,\delta) = p_i,
\end{equation*}

where $G(y;\lambda,\beta,\delta)$ denotes the cumulative distribution function of the NTLE distribution. The percentile estimators are obtained by minimizing the sum of squared differences

\begin{equation*}
\sum_{i=1}^{n}
\left[
G(y_{(i)};\lambda,\beta,\delta) - p_i
\right]^2 .
\end{equation*}

The resulting estimates of the parameters are obtained using numerical optimization procedures.









\subsection{Maximum Goodness-of-Fit Estimation}

The maximum goodness-of-fit estimation (MGFE) approach estimates the parameters by minimizing the maximum absolute deviation between the fitted distribution and the empirical distribution \citep{Luceno2006}

The MGFE objective function is defined as

\begin{equation}
M(\lambda,\beta,\delta)
=
\max_{1\le i\le n}
\left|
G(y_{(i)})-\frac{i-0.5}{n}
\right|.
\end{equation}

Substituting the NTLE cumulative distribution function yields

\begin{equation}
M(\lambda,\beta,\delta)
=
\max_{1\le i\le n}
\left|
\frac{(e^{\lambda y_{(i)}}-1)^\beta\left[1+\delta+(e^{\lambda y_{(i)}}-1)^\beta\right]}
{\left[1+(e^{\lambda y_{(i)}}-1)^\beta\right]^2}
-
\frac{i-0.5}{n}
\right|.
\end{equation}

The MGFE estimators $(\hat{\lambda}_{MGFE},\hat{\beta}_{MGFE},\hat{\delta}_{MGFE})$ are obtained by minimizing $M(\lambda,\beta,\delta)$ numerically.

\subsection{Bayesian Estimation}

In the Bayesian framework, the parameters of the NTLE distribution are treated as random variables with prior distributions \citep{gelman2014bayesian, dutta2024bayesian}.

Let the joint prior distribution of the parameter vector be
\[
\pi(\lambda,\beta,\delta).
\]
After observing the data $y$, the posterior distribution is obtained using Bayes' theorem:

\begin{equation*}
\pi(\lambda,\beta,\delta|y)
\propto
L(\lambda,\beta,\delta)\pi(\lambda,\beta,\delta).
\end{equation*}

The Bayesian estimates are derived from the posterior distribution. Under the squared error loss function, the Bayesian estimator of each parameter is given by the posterior mean
\begin{equation}
\hat{\theta}_{BAYES} = e(\theta|y).
\end{equation}

In most cases, the posterior distribution may not have a closed-form expression; Bayesian estimates are often obtained using Markov chain Monte Carlo (MCMC) methods.

\section{Simulation Study}

In this section, we conduct Monte Carlo simulation to investigate the finite–sample performance of the parameter estimators discussed in Section 3 for the New Transmuted Logistic exponential (NTLE) distribution. The objective of the simulation experiment was to assess the relative efficiency and stability of the competing estimation methods under different sample sizes and several 
parameter values

Random samples were generated from the NTLE distribution with parameter vector  $\boldsymbol{\theta}=(\lambda,\beta,\delta)$ via the inverse transform approach \citep{adesegun2023tle}. The parameter values were fixed at  $ (\lambda=1, \beta=1.5,\delta=0.5) $ and $(\lambda=0.5, \beta=1.5,\delta=0.2)$. The simulation experiment was carried out for the sample sizes $ n = 20, 50,100, 200, 500$ and  $ 1000.$. For each generated sample, the parameters $(\lambda,\beta,\delta)$ were estimated using the following ten estimation techniques: MLE, MME, LSE, WLSE, MPS, BAYES, ADE, CVME, PCE, and MGFE. Note that for the Bayes method, independent Gamma prior distributions are assumed for the parameters $\lambda$ and $\beta$, while a transformation on $\delta$ implicitly restricts it to the interval $(-1,1)$. The posterior sampling is carried out using a random--walk Metropolis--Hastings.  To evaluate the performance of the estimators, the bias, mean squared error and root mean square error were considered.

\subsection{Evaluation Metrics}

\subsection*{Bias} Bias measures the average deviation of the estimator from the true parameter value. The bias of an estimator $\hat{\theta}$ is defined as
\[
Bias(\hat{\theta}) =
\frac{1}{R}\sum_{r=1}^{R}
(\hat{\theta}_r-\theta),
\]

\noindent where $\hat{\theta}_r$ denotes the estimate obtained from the $r$th replication and $\theta$ represents the true parameter value.

\subsection*{Mean Squared Error}

The mean squared Error (MSE) is defined as
\[
MSE(\hat{\theta}) =
\frac{1}{R}\sum_{r=1}^{R}
(\hat{\theta}_r-\theta)^2.
\]

The MSE combines both the variance and the bias of an estimator and provides a useful measure of overall estimation accuracy.

\subsection*{Root Mean Squared error}

This is simply the square root of the mean squared error. It is given as:

\[
RMSE(\hat{\theta}) = \sqrt{
\frac{1}{R}\sum_{r=1}^{R}
(\hat{\theta}_r-\theta)^2.}
\]

\subsection{Simulation Results}
\noindent The results of the simulation study provide a comprehensive comparison of the estimators in terms of bias and mean squared error across different sample sizes. These findings offer useful insights into the relative efficiency and robustness of the competing estimation procedures for the NTLE distribution.

\begin{table}[h]
\centering
\caption{Bias, MSE and RMSE of the estimators for the NTLE parameters when $n=20$ and $(\lambda=1, \beta=1.5,\delta=0.5)$}
\label{tab:sim_20}
\begin{tabular}{lccccccccc}
\hline
 & \multicolumn{3}{c}{Bias} & \multicolumn{3}{c}{MSE} & \multicolumn{3}{c}{RMSE} \\
Method & $\hat{\lambda}$ & $\hat{\beta}$ & $\hat{\delta}$ &
$\hat{\lambda}$ & $\hat{\beta}$ & $\hat{\delta}$ &
$\hat{\lambda}$ & $\hat{\beta}$ & $\hat{\delta}$ \\
\hline
MLE   & 0.9523 & -0.2396 & -0.9278 & 1.8805 & 0.1747 & 1.2133 & 1.3713 & 0.4180 & 1.1015 \\
MME   & 0.5128 & -0.0782 & -0.5589 & 0.6015 & 0.1629 & 0.8963 & 0.7756 & 0.4036 & 0.9467 \\
LSE   & 0.7155 & -0.1873 & -0.6913 & 1.0864 & 0.4108 & 0.9056 & 1.0423 & 0.6409 & 0.9517 \\
WLSE   & 0.7166 & -0.2381 & -0.6127 & 1.4192 & 0.1922 & 0.8350 & 1.1913 & 0.4384 & 0.9138 \\
MPS   & 0.9855 & -0.3886 & -0.7462 & 2.6304 & 0.2702 & 1.0398 & 1.6218 & 0.5198 & 1.0197 \\
BAYES   & 0.5530 & -0.1292 & -0.5608 & 0.4002 & 0.0693 & 0.3767 & 0.6326 & 0.2632 & 0.6138 \\
ADE   & 0.6979 & -0.2137 & -0.6355 & 1.2973 & 0.1692 & 0.8258 & 1.1390 & 0.4114 & 0.9088 \\
CVME   & 0.7182 & -0.0711 & -0.7734 & 0.9302 & 0.4874 & 0.9764 & 0.9644 & 0.6982 & 0.9881 \\
PCE   & 0.7157 & -0.3867 & -0.5929 & 1.3915 & 0.2397 & 0.8176 & 1.1796 & 0.4896 & 0.9042 \\
MGFE   & 0.3480 & 0.0234 & -0.4022 & 0.2181 & 0.2015 & 0.2780 & 0.4670 & 0.4489 & 0.5273 \\
\hline
\end{tabular}
\end{table}

\begin{table}[h]
\centering
\caption{Bias, MSE and RMSE of the estimators for the NTLE parameters when $n=50$ and $(\lambda=1, \beta=1.5,\delta=0.5)$.}
\label{tab:sim_50}
\begin{tabular}{lccccccccc}
\hline
 & \multicolumn{3}{c}{Bias} & \multicolumn{3}{c}{MSE} & \multicolumn{3}{c}{RMSE} \\
Method & $\hat{\lambda}$ & $\hat{\beta}$ & $\hat{\delta}$ &
$\hat{\lambda}$ & $\hat{\beta}$ & $\hat{\delta}$ &
$\hat{\lambda}$ & $\hat{\beta}$ & $\hat{\delta}$ \\
\hline
MLE   & 0.2539 & -0.0637 & -0.3578 & 0.2861 & 0.0339 & 0.3930 & 0.5349 & 0.1840 & 0.6269 \\
MME   & 0.3773 & -0.1768 & -0.5463 & 0.3755 & 0.1218 & 0.7507 & 0.6128 & 0.3489 & 0.8664 \\
LSE   & 0.2425 & -0.0918 & -0.3021 & 0.2594 & 0.0667 & 0.3894 & 0.5093 & 0.2583 & 0.6240 \\
WLSE   & 0.1986 & -0.0722 & -0.2558 & 0.2090 & 0.0512 & 0.3240 & 0.4572 & 0.2263 & 0.5692 \\
MPS   & 0.1933 & -0.1520 & -0.2556 & 0.2730 & 0.0542 & 0.3556 & 0.5225 & 0.2329 & 0.5963 \\
BAYES   & 0.4203 & -0.1301 & -0.6022 & 0.2290 & 0.0431 & 0.4324 & 0.4786 & 0.2077 & 0.6576 \\
ADE   & 0.0956 & -0.0378 & -0.1203 & 0.0747 & 0.0435 & 0.1824 & 0.2733 & 0.2086 & 0.4271 \\
CVME   & 0.2492 & -0.0474 & -0.3112 & 0.2515 & 0.0651 & 0.3996 & 0.5015 & 0.2552 & 0.6322 \\
PCE   & 0.3623 & -0.3230 & -0.4232 & 0.4375 & 0.1596 & 0.5526 & 0.6615 & 0.3996 & 0.7433 \\
MGFE  & 0.2506 & -0.0195 & -0.4074 & 0.1218 & 0.0436 & 0.3322 & 0.3491 & 0.2088 & 0.5764 \\
\hline
\end{tabular}
\end{table}

\begin{table}[h]
\centering
\caption{Bias, MSE and RMSE of the estimators for the NTLE parameter when $n=100$ and $(\lambda=1, \beta=1.5,\delta=0.5)$}
\label{tab:sim_100}
\begin{tabular}{lccccccccc}
\hline
 & \multicolumn{3}{c}{Bias} & \multicolumn{3}{c}{MSE} & \multicolumn{3}{c}{RMSE} \\
Method & $\hat{\lambda}$ & $\hat{\beta}$ & $\hat{\delta}$ &
$\hat{\lambda}$ & $\hat{\beta}$ & $\hat{\delta}$ &
$\hat{\lambda}$ & $\hat{\beta}$ & $\hat{\delta}$ \\
\hline
MLE   & 0.3204 & -0.1300 & -0.5469 & 0.2461 & 0.0479 & 0.6101 & 0.4961 & 0.2189 & 0.7811 \\
MME   & 0.3934 & -0.1885 & -0.6476 & 0.3629 & 0.0873 & 0.8210 & 0.6024 & 0.2954 & 0.9061 \\
LSE   & 0.3062 & -0.1158 & -0.4653 & 0.3251 & 0.0718 & 0.5430 & 0.5702 & 0.2680 & 0.7369 \\
WLSE   & 0.3324 & -0.1376 & -0.5286 & 0.3219 & 0.0656 & 0.6133 & 0.5674 & 0.2561 & 0.7831 \\
MPS   & 0.3341 & -0.2081 & -0.4899 & 0.3638 & 0.0857 & 0.6845 & 0.6032 & 0.2928 & 0.8273 \\
BAYES   & 0.3675 & -0.1484 & -0.6163 & 0.1675 & 0.0346 & 0.4717 & 0.4093 & 0.1860 & 0.6868 \\
ADE   & 0.4046 & -0.1777 & -0.6214 & 0.3887 & 0.0808 & 0.7736 & 0.6235 & 0.2842 & 0.8795 \\
CVME   & 0.3385 & -0.0969 & -0.5391 & 0.3197 & 0.0681 & 0.5768 & 0.5654 & 0.2609 & 0.7595 \\
PCE   & 0.1281 & -0.1890 & -0.2247 & 0.1581 & 0.0615 & 0.3490 & 0.3977 & 0.2480 & 0.5908 \\
MGFE   & 0.2993 & -0.0844 & -0.5713 & 0.1411 & 0.0281 & 0.5283 & 0.3757 & 0.1675 & 0.7269 \\
\hline
\end{tabular}
\end{table}

\begin{table}[h]
\centering
\caption{Bias, MSE and RMSE of the estimators for the NTLE parameters when when $n=200$ and $(\lambda=1, \beta=1.5,\delta=0.5)$}
\label{tab:sim_200}
\begin{tabular}{lccccccccc}
\hline
 & \multicolumn{3}{c}{Bias} & \multicolumn{3}{c}{MSE} & \multicolumn{3}{c}{RMSE} \\
Method & $\hat{\lambda}$ & $\hat{\beta}$ & $\hat{\delta}$ &
$\hat{\lambda}$ & $\hat{\beta}$ & $\hat{\delta}$ &
$\hat{\lambda}$ & $\hat{\beta}$ & $\hat{\delta}$ \\
\hline
MLE   & 0.2676 & -0.1233 & -0.3961 & 0.2190 & 0.0465 & 0.4367 & 0.468 & 0.2156 & 0.6609 \\
MME   & 0.1754 & -0.1591 & -0.2117 & 0.2200 & 0.0549 & 0.3769 & 0.4690 & 0.2344 & 0.614 \\
LSE   & 0.3240 & -0.1266 & -0.4566 & 0.2948 & 0.0595 & 0.4690 & 0.5429 & 0.2440 & 0.6849 \\
WLSE   & 0.3801 & -0.1623 & -0.4876 & 0.4297 & 0.0824 & 0.5946 & 0.6555 & 0.2870 & 0.7711 \\
MPS   & 0.1945 & -0.1606 & -0.2492 & 0.2205 & 0.0592 & 0.4233 & 0.4696 & 0.2432 & 0.6506 \\
BAYES   & 0.3366 & -0.1559 & -0.4939 & 0.1939 & 0.0418 & 0.3850 & 0.4404 & 0.2046 & 0.6205 \\
ADE   & 0.3944 & -0.1626 & -0.5104 & 0.4366 & 0.0834 & 0.6068 & 0.6608 & 0.2888 & 0.7790 \\
CVME   & 0.4405 & -0.1619 & -0.5599 & 0.4957 & 0.0930 & 0.6652 & 0.7041 & 0.3050 & 0.8156 \\
PCE   & -0.0022 & -0.1621 & 0.0322 & 0.0673 & 0.0439 & 0.1628 & 0.2595 & 0.2096 & 0.4035 \\
MGFE   & 0.2869 & -0.0807 & -0.4880 & 0.1289 & 0.0272 & 0.3477 & 0.3591 & 0.1649 & 0.5897 \\
\hline
\end{tabular}
\end{table}

\begin{table}[h]
\centering
\caption{Bias, MSE and RMSE of the estimators for the NTLE parameters when $n=500$ and $(\lambda=1, \beta=1.5,\delta=0.5)$}
\label{tab:sim_500}
\begin{tabular}{lccccccccc}
\hline
 & \multicolumn{3}{c}{Bias} & \multicolumn{3}{c}{MSE} & \multicolumn{3}{c}{RMSE} \\
Method & $\hat{\lambda}$ & $\hat{\beta}$ & $\hat{\delta}$ &
$\hat{\lambda}$ & $\hat{\beta}$ & $\hat{\delta}$ &
$\hat{\lambda}$ & $\hat{\beta}$ & $\hat{\delta}$ \\
\hline
MLE   & 0.1334 & -0.0472 & -0.2224 & 0.0775 & 0.0116 & 0.1896 & 0.2784 & 0.1078 & 0.4354 \\
MME   & 0.2363 & -0.1403 & -0.3594 & 0.1748 & 0.0384 & 0.3838 & 0.4180 & 0.1959 & 0.6195 \\
LSE   & 0.3124 & -0.1058 & -0.4661 & 0.2443 & 0.0368 & 0.4432 & 0.4943 & 0.1918 & 0.6657 \\
WLSE   & 0.2031 & -0.0630 & -0.3345 & 0.1137 & 0.0159 & 0.2584 & 0.3372 & 0.1262 & 0.5083 \\
MPS   & 0.0532 & -0.0678 & -0.0641 & 0.0739 & 0.0132 & 0.1757 & 0.2718 & 0.1150 & 0.4191 \\
BAYES   & 0.2705 & -0.1031 & -0.4236 & 0.1193 & 0.0196 & 0.2800 & 0.3454 & 0.1398 & 0.5292 \\
ADE   & 0.2546 & -0.0876 & -0.3891 & 0.1843 & 0.0298 & 0.3621 & 0.4293 & 0.1725 & 0.6017 \\
CVME   & 0.3685 & -0.1248 & -0.5205 & 0.3289 & 0.0528 & 0.5423 & 0.5735 & 0.2298 & 0.7364 \\
PCE   & 0.0251 & -0.1083 & -0.0038 & 0.0759 & 0.0250 & 0.1547 & 0.2755 & 0.1582 & 0.3933 \\
MGFE   & 0.2464 & -0.0597 & -0.4204 & 0.1059 & 0.0150 & 0.2830 & 0.3254 & 0.1225 & 0.5319 \\
\hline
\end{tabular}
\end{table}

\begin{table}[h]
\centering
\caption{Bias, MSE and RMSE of the estimators for the NTLE parameters when $n=1000$ and $(\lambda=1, \beta=1.5,\delta=0.5)$.}
\label{tab:sim_1000}
\begin{tabular}{lccccccccc}
\hline
 & \multicolumn{3}{c}{Bias} & \multicolumn{3}{c}{MSE} & \multicolumn{3}{c}{RMSE} \\
Method & $\hat{\lambda}$ & $\hat{\beta}$ & $\hat{\delta}$ &
$\hat{\lambda}$ & $\hat{\beta}$ & $\hat{\delta}$ &
$\hat{\lambda}$ & $\hat{\beta}$ & $\hat{\delta}$ \\
\hline
MLE   & 0.0708 & -0.0129 & -0.1324 & 0.0229 & 0.0035 & 0.0865 & 0.1514 & 0.0594 & 0.2941 \\
MME   & 0.1657 & -0.1040 & -0.2646 & 0.1109 & 0.0222 & 0.3024 & 0.3331 & 0.1491 & 0.5499 \\
LSE   & 0.3640 & -0.1139 & -0.5297 & 0.3114 & 0.0465 & 0.5497 & 0.5580 & 0.2157 & 0.7414 \\
WLSE   & 0.0909 & -0.0187 & -0.1604 & 0.0340 & 0.0046 & 0.1139 & 0.1843 & 0.0677 & 0.3374 \\
MPS   & 0.0337 & -0.0397 & -0.0339 & 0.0553 & 0.0122 & 0.1299 & 0.2352 & 0.1104 & 0.3604 \\
BAYES   & 0.1540 & -0.0455 & -0.2615 & 0.0586 & 0.0082 & 0.1562 & 0.2422 & 0.0903 & 0.3952 \\
ADE   & 0.0771 & -0.0198 & -0.1340 & 0.0330 & 0.0046 & 0.1102 & 0.1817 & 0.0681 & 0.3320 \\
CVME   & 0.3641 & -0.1114 & -0.5305 & 0.3106 & 0.0461 & 0.5493 & 0.5573 & 0.2146 & 0.7412 \\
PCE   & 0.0562 & -0.1004 & -0.0593 & 0.0908 & 0.0249 & 0.1913 & 0.3013 & 0.1576 & 0.4374 \\
MGFE   & 0.2374 & -0.0287 & -0.4304 & 0.0818 & 0.0084 & 0.2576 & 0.2860 & 0.0915 & 0.5076 \\

\hline
\end{tabular}
\end{table}

\begin{table}[h]
\centering
\caption{Bias, MSE and RMSE of the estimators for the NTLE parameters when $n=20$ and $(\lambda=0.5, \beta=1.5,\delta=0.2)$}
\label{tab:sim_20a}
\begin{tabular}{lccccccccc}
\hline
 & \multicolumn{3}{c}{Bias} & \multicolumn{3}{c}{MSE} & \multicolumn{3}{c}{RMSE} \\
Method & $\hat{\lambda}$ & $\hat{\beta}$ & $\hat{\delta}$ &
$\hat{\lambda}$ & $\hat{\beta}$ & $\hat{\delta}$ &
$\hat{\lambda}$ & $\hat{\beta}$ & $\hat{\delta}$ \\
\hline
MLE   & 0.2636 & -0.1702 & -0.4939 & 0.2195 & 0.1199 & 0.5446 & 0.4685 & 0.3463 & 0.7380 \\
MME   & 0.1999 & -0.0444 & -0.4571 & 0.0966 & 0.1894 & 0.6232 & 0.3109 & 0.4352 & 0.7894 \\
LSE   & 0.2486 & -0.1931 & -0.4041 & 0.1687 & 0.4080 & 0.5820 & 0.4108 & 0.6387 & 0.7629 \\
WLSE   & 0.2464 & -0.2403 & -0.3180 & 0.2306 & 0.1925 & 0.5578 & 0.4802 & 0.4387 & 0.7469 \\
MPS   & 0.3086 & -0.3463 & -0.3547 & 0.3930 & 0.2333 & 0.5642 & 0.6269 & 0.4830 & 0.7512 \\
BAYES   & 0.2627 & -0.1830 & -0.3907 & 0.0984 & 0.0937 & 0.2087 & 0.3136 & 0.3060 & 0.4569 \\
ADE   & 0.2584 & -0.2402 & -0.4145 & 0.2107 & 0.1756 & 0.6010 & 0.4591 & 0.4191 & 0.7752 \\
CVME   & 0.2950 & -0.1123 & -0.5694 & 0.1828 & 0.5034 & 0.706 & 0.4275 & 0.7095 & 0.8402 \\
PCE   & 0.2658 & -0.4007 & -0.3664 & 0.2270 & 0.2615 & 0.5640 & 0.4765 & 0.5113 & 0.7510 \\
MGFE   & 0.1117 & -0.0060 & -0.2107 & 0.0266 & 0.1809 & 0.1674 & 0.1631 & 0.4253 & 0.4091 \\
\hline
\end{tabular}
\end{table}

\begin{table}[h]
\centering
\caption{Bias, MSE and RMSE of the estimators for the NTLE parameters when $n=50$ and $(\lambda=0.5, \beta=1.5,\delta=0.2)$}
\label{tab:sim_50b}
\begin{tabular}{lccccccccc}
\hline
 & \multicolumn{3}{c}{Bias} & \multicolumn{3}{c}{MSE} & \multicolumn{3}{c}{RMSE} \\
Method & $\hat{\lambda}$ & $\hat{\beta}$ & $\hat{\delta}$ &
$\hat{\lambda}$ & $\hat{\beta}$ & $\hat{\delta}$ &
$\hat{\lambda}$ & $\hat{\beta}$ & $\hat{\delta}$ \\
\hline
MLE   & 0.0761 & -0.0826 & -0.1389 & 0.0610 & 0.0471 & 0.3322 & 0.2469 & 0.2170 & 0.5763 \\ 
MME   & 0.1002 & -0.1487 & -0.2653 & 0.0488 & 0.1106 & 0.4974 & 0.2209 & 0.3325 & 0.7053 \\
LSE   & 0.0487 & -0.0980 & -0.0519 & 0.0386 & 0.0666 & 0.2927 & 0.1964 & 0.2580 & 0.5410 \\
WLSE   & 0.0209 & -0.0731 & 0.0391 & 0.0322 & 0.0504 & 0.2581 & 0.1795 & 0.2244 & 0.5080 \\ 
MPS   & 0.0235 & -0.1459 & 0.0257 & 0.0444 & 0.0533 & 0.2883 & 0.2107 & 0.2309 & 0.5370 \\
BAYES   & 0.1766 & -0.1772 & -0.4550 & 0.0441 & 0.0603 & 0.2734 & 0.2099 & 0.2457 & 0.5229 \\ 
ADE   & 0.0185 & -0.0677 & 0.0398 & 0.0306 & 0.0447 & 0.2567 & 0.1749 & 0.2115 & 0.5067 \\
CVME   & 0.0729 & -0.0597 & -0.1605 & 0.0359 & 0.0641 & 0.2858 & 0.1894 & 0.2531 & 0.5346 \\ 
PCE   & 0.0641 & -0.3044 & -0.0466 & 0.0642 & 0.1507 & 0.4208 & 0.2534 & 0.3882 & 0.6487 \\
MGFE   & 0.0665 & -0.0373 & -0.1823 & 0.0223 & 0.0398 & 0.1989 & 0.1492 & 0.1996 & 0.4460 \\
\hline
\end{tabular}
\end{table}

\begin{table}[h]
\centering
\caption{Bias, MSE and RMSE of the estimators for the NTLE parameter when $n=100$ and $(\lambda=0.5, \beta=1.5,\delta=0.2)$}
\label{tab:sim_100c}
\begin{tabular}{lccccccccc}
\hline
 & \multicolumn{3}{c}{Bias} & \multicolumn{3}{c}{MSE} & \multicolumn{3}{c}{RMSE} \\
Method & $\hat{\lambda}$ & $\hat{\beta}$ & $\hat{\delta}$ &
$\hat{\lambda}$ & $\hat{\beta}$ & $\hat{\delta}$ &
$\hat{\lambda}$ & $\hat{\beta}$ & $\hat{\delta}$ \\
\hline
MLE   & 0.0431 & -0.1028 & -0.1110 & 0.0264 & 0.0261 & 0.2409 & 0.1626 & 0.1616 & 0.4908 \\
MME   & 0.1418 & -0.2956 & -0.2914 & 0.0931 & 0.1625 & 0.6875 & 0.3052 & 0.4031 & 0.8291 \\ 
LSE   & 0.1407 & -0.2140 & -0.3596 & 0.0695 & 0.0935 & 0.4400 & 0.2636 & 0.3059 & 0.6633 \\
WLSE   & 0.0815 & -0.1569 & -0.1864 & 0.0515 & 0.0607 & 0.3283 & 0.2269 & 0.2464 & 0.5729 \\
MPS   & 0.0199 & -0.1578 & -0.0074 & 0.0316 & 0.0400 & 0.2979 & 0.1779 & 0.2000 & 0.5458 \\
BAYES   & 0.1063 & -0.1627 & -0.2455 & 0.0319 & 0.0462 & 0.1504 & 0.1785 & 0.2148 & 0.3878 \\
ADE   & 0.1919 & -0.2433 & -0.4322 & 0.1220 & 0.1293 & 0.5801 & 0.3493 & 0.3595 & 0.7616 \\
CVME   & 0.1633 & -0.2195 & -0.4389 & 0.0748 & 0.0920 & 0.5096 & 0.2735 & 0.3033 & 0.7138 \\ 
PCE   & 0.0305 & -0.2373 & 0.0365 & 0.0571 & 0.1014 & 0.3331 & 0.2389 & 0.3184 & 0.5771 \\
MGFE   & 0.0681 & -0.1066 & -0.2184 & 0.0220 & 0.0376 & 0.1808 & 0.1484 & 0.1940 & 0.4252 \\
\hline
\end{tabular}
\end{table}

\begin{table}[h]
\centering
\caption{Bias, MSE and RMSE of the estimators for the NTLE parameters when $n=200$ and $(\lambda=0.5, \beta=1.5,\delta=0.2)$}
\label{tab:sim_200d}
\begin{tabular}{lccccccccc}
\hline
 & \multicolumn{3}{c}{Bias} & \multicolumn{3}{c}{MSE} & \multicolumn{3}{c}{RMSE} \\
Method & $\hat{\lambda}$ & $\hat{\beta}$ & $\hat{\delta}$ &
$\hat{\lambda}$ & $\hat{\beta}$ & $\hat{\delta}$ &
$\hat{\lambda}$ & $\hat{\beta}$ & $\hat{\delta}$ \\
\hline
MLE & 0.0772 & -0.1297 & -0.2010 & 0.0319 & 0.0495 & 0.2897 & 0.1787 & 0.2225 & 0.5383 \\
MME & 0.0985 & -0.1728 & -0.2131 & 0.0582 & 0.0726 & 0.4008 & 0.2411 & 0.2694 & 0.6331 \\
LSE & 0.0688 & -0.1152 & -0.1849 & 0.0230 & 0.0452 & 0.2191 & 0.1518 & 0.2127 & 0.4681 \\
WLSE & 0.0555 & -0.1164 & -0.1281 & 0.0255 & 0.0459 & 0.2500 & 0.1596 & 0.2143 & 0.5000 \\
MPS & 0.0760 & -0.1915 & -0.1236 & 0.0508 & 0.0825 & 0.4317 & 0.2253 & 0.2872 & 0.6570 \\
BAYES & 0.1230 & -0.1796 & -0.3221 & 0.0346 & 0.0544 & 0.2508 & 0.1860 & 0.2332 & 0.5008 \\
ADE & 0.1335 & -0.1882 & -0.2806 & 0.0774 & 0.1009 & 0.4594 & 0.2781 & 0.3177 & 0.6778 \\
CVME & 0.0694 & -0.1034 & -0.1893 & 0.0224 & 0.0429 & 0.2162 & 0.1498 & 0.2072 & 0.4649 \\
PCE & -0.0427 & -0.1391 & 0.2286 & 0.0163 & 0.0340 & 0.2596 & 0.1276 & 0.1845 & 0.5095 \\
MGFE & 0.0669 & -0.0931 & -0.2255 & 0.0132 & 0.0283 & 0.1657 & 0.1150 & 0.1682 & 0.4071 \\
\hline
\end{tabular}
\end{table}

\begin{table}[h]
\centering
\caption{Bias, MSE and RMSE of the estimators for the NTLE parameters when $n=500$ and $(\lambda=0.5, \beta=1.5,\delta=0.2)$}
\label{tab:sim_500e}
\begin{tabular}{lccccccccc}
\hline
 & \multicolumn{3}{c}{Bias} & \multicolumn{3}{c}{MSE} & \multicolumn{3}{c}{RMSE} \\
Method & $\hat{\lambda}$ & $\hat{\beta}$ & $\hat{\delta}$ &
$\hat{\lambda}$ & $\hat{\beta}$ & $\hat{\delta}$ &
$\hat{\lambda}$ & $\hat{\beta}$ & $\hat{\delta}$ \\
\hline
MLE   & 0.0672 & -0.0907 & -0.2069 & 0.0209 & 0.0226 & 0.2093 & 0.1445 & 0.1503 & 0.4575 \\
MME   & 0.1183 & -0.1915 & -0.3054 & 0.0518 & 0.0688 & 0.4254 & 0.2277 & 0.2624 & 0.6522 \\
LSE   & 0.0983 & -0.1351 & -0.2362 & 0.0467 & 0.0562 & 0.3239 & 0.2160 & 0.2370 & 0.5692 \\
WLSE   & 0.0379 & -0.0681 & -0.1005 & 0.0145 & 0.0177 & 0.1449 & 0.1206 & 0.1331 & 0.3807 \\
MPS   & 0.0349 & -0.1007 & -0.0670 & 0.0220 & 0.0244 & 0.2486 & 0.1483 & 0.1563 & 0.4986 \\
BAYES   & 0.0981 & -0.1347 & -0.2710 & 0.0246 & 0.0310 & 0.2067 & 0.1569 & 0.1761 & 0.4546 \\
ADE   & 0.0730 & -0.1054 & -0.1846 & 0.0310 & 0.0402 & 0.2503 & 0.1762 & 0.2005 & 0.5003 \\ 
CVME   & 0.1257 & -0.1558 & -0.3030 & 0.0595 & 0.0723 & 0.3933 & 0.2440 & 0.2690 & 0.6272 \\
PCE   & -0.0149 & -0.1024 & 0.1386 & 0.0213 & 0.0264 & 0.2455 & 0.1458 & 0.1623 & 0.4955 \\
MGFE   & 0.0532 & -0.0791 & -0.1487 & 0.0165 & 0.0248 & 0.1674 & 0.1283 & 0.1573 & 0.4092 \\
\hline
\end{tabular}
\end{table}

\begin{table}[h]
\centering
\caption{Bias, MSE and RMSE of the estimators for the NTLE parameters when $n=1000$ and $(\lambda=0.5, \beta=1.5,\delta=0.2)$}
\label{tab:sim_1000f}
\begin{tabular}{lccccccccc}
\hline
 & \multicolumn{3}{c}{Bias} & \multicolumn{3}{c}{MSE} & \multicolumn{3}{c}{RMSE} \\
Method & $\hat{\lambda}$ & $\hat{\beta}$ & $\hat{\delta}$ &
$\hat{\lambda}$ & $\hat{\beta}$ & $\hat{\delta}$ &
$\hat{\lambda}$ & $\hat{\beta}$ & $\hat{\delta}$ \\
\hline
MLE   & 0.0743 & -0.0992 & -0.2027 & 0.0304 & 0.0385 & 0.2852 & 0.1744 & 0.1962 & 0.5341 \\
MME   & 0.0674 & -0.1017 & -0.2122 & 0.0216 & 0.0257 & 0.2289 & 0.1468 & 0.1604 & 0.4784 \\
LSE   & 0.0962 & -0.1218 & -0.2165 & 0.0508 & 0.0627 & 0.3631 & 0.2254 & 0.2503 & 0.6026 \\
WLSE   & 0.0623 & -0.0917 & -0.1395 & 0.0341 & 0.0397 & 0.2788 & 0.1846 & 0.1992 & 0.5280 \\
MPS   & 0.0831 & -0.1356 & -0.1996 & 0.0407 & 0.0518 & 0.3851 & 0.2019 & 0.2276 & 0.6205 \\
BAYES   & 0.0677 & -0.1007 & -0.1845 & 0.0200 & 0.0250 & 0.1849 & 0.1415 & 0.1581 & 0.4301 \\
ADE   & 0.0639 & -0.0909 & -0.1281 & 0.0390 & 0.0460 & 0.2838 & 0.1976 & 0.2145 & 0.5327 \\
CVME   & 0.1123 & -0.1269 & -0.2805 & 0.0520 & 0.0634 & 0.3669 & 0.2280 & 0.2519 & 0.6058 \\
PCE   & 0.0038 & -0.0819 & 0.0818 & 0.0292 & 0.0308 & 0.2635 & 0.1708 & 0.1754 & 0.5133 \\
MGFE   & 0.0619 & -0.0650 & -0.1925 & 0.0171 & 0.0242 & 0.1745 & 0.1310 & 0.1557 & 0.4177 \\
\hline
\end{tabular}
\end{table}

The simulation results in Tables~\ref{tab:sim_20}--\ref{tab:sim_1000f} show a clear improvement in estimator performance as the sample size increases. For both parameter settings, the absolute bias, MSE, and RMSE generally decline with increasing $n$, indicating that all methods benefit from larger samples.

For the first setting, $(\lambda,\beta,\delta)=(1,1.5,0.5)$, the small-sample case in Table~\ref{tab:sim_20} shows noticeable bias for most methods, especially in estimating $\delta$. among the competing estimators, MGFE performs particularly well at $n=20$, giving the smallest MSE and RMSE for $\hat{\lambda}$ and $\hat{\delta}$, while the Bayesian method performs best for $\hat{\beta}$. as the sample size increases (Tables~\ref{tab:sim_50}--\ref{tab:sim_1000}), the estimators become more stable. In the larger samples, MLE, MPS, PCE, and MGFE are generally the most competitive, with MPS and PCE showing especially strong performance at $n=500$ and $n=1000$.

A similar pattern is observed for the second setting, $(\lambda,\beta,\delta)=(0.5,1.5,0.2)$, reported in Tables~\ref{tab:sim_20a}--\ref{tab:sim_1000f}. at $n=20$, MGFE again stands out, particularly for $\hat{\lambda}$ and $\hat{\delta}$, while the Bayesian estimator remains competitive. For moderate and large sample sizes, no single method dominates uniformly across all three parameters, but MGFE, WLSE, MLE, and Bayesian estimation repeatedly show good overall performance. In particular, MGFE gives some of the smallest MSE and RMSE values in Tables~\ref{tab:sim_50b}, \ref{tab:sim_200d}, \ref{tab:sim_500e}, and \ref{tab:sim_1000f}.

\clearpage
\section{Application with infectious diseases death data}

In this section, the practical usefulness of the proposed New Transmuted Logistic-exponential (NTLE) distribution is illustrated by means of a real dataset. The data used is the deaths due to COVID-19 in Egypt in 2020 \citep{el2021exponentiated}.  The main objective is to assess the ability of the NTLE model to describe the observed data and to compare its performance with its base distributions, namely, the exponential and logistic-exponential (Le) distribution.

\subsection{Comparison of NTLE with base distributions}
Here, for the mortality data, the NTLE parameters were MLE and MGFE estimation methods, while the competing models were fitted by maximum likelihood estimation. To examine the relative adequacy of the fitted models, the following criteria were considered: the log-likelihood value, Akaike information criterion (AIC), Bayesian information criterion (BIC), corrected Akaike information criterion (CAIC), Hannan-Quinn information criterion (HQIC), the Kolmogorov-Smirnov (K--S) statistic, and the corresponding $p$-value.

Let $\ell(\hat{\Phi})$ denote the maximized log-likelihood, $p$ the number of estimated parameters, and $m$ the sample size. Then, the AIC, BIC, CAIC, HQIC, and K-S statistics are defined, respectively as: 
\[
AIC=-2\ell(\hat{\Phi})+2p.
\]
\[
BIC=-2\ell(\hat{\Phi})+p\log m.
\]
\[
CAIC=-2\ell(\hat{\Phi})+p(\log m+1).
\]
\[
HQIC=-2\ell(\hat{\Phi})+2p\log(\log m).
\]
\[
D_m=\sup_y \left|G_m(y)-G(y;\hat{\Phi})\right|,
\]
where $G_m(y)$ is the empirical distribution function and $G(y;\hat{\theta})$ is the fitted cumulative distribution function. A smaller value of $D_m$ indicates closer agreement between the fitted model and the observed data.

The corresponding $p$-value is obtained from the Kolmogorov-Smirnov test and is used to assess the adequacy of the fitted distribution. A larger $p$-value suggests that the fitted model is more consistent with the observed data. The fitted models were also examined graphically using the histogram with fitted density curves and the empirical cumulative distribution function plotted against the fitted cumulative distribution functions.

\begin{sidewaystable}[htbp]
\centering
\caption{parameter estimates and goodness-of-fit measures for the fitted models.}
\label{tab:model_comparison}
\begin{tabular}{lccccccccccc}
\hline
Model &  $\hat{\alpha}$ & $\hat{\beta}$  & $\hat{\lambda}$ & loglik & AIC & BIC & CAIC & HQIC & K-S & $p$-value \\
\hline
exponential   & 0.0036 &  &  & -510.1720 & 1022.3440 & 1024.6878 & 1025.6878 & 1023.2815 & 0.1408 & 0.0945 \\
Le   & 0.0050 & 0.6571 &  & -502.8429 & 1009.6857 & 1014.3733 & 1016.3733 & 1011.5607 & 0.1286 & 0.1564 \\
NTLE (MLE)   & 0.0085 & 0.4732 & -0.5830 & -500.4086 & 1006.8171 & 1013.8485 & 1016.8485 & 1009.6296 & 0.0967 & 0.4669 \\
NTLE (MGF)   & 0.0134 & 0.2935 & -0.9272 & -500.7561 & 1007.5122 & 1014.5436 & 1017.5436 & 1010.3247 & 0.0791 & 0.7216 \\
\hline
\end{tabular}
\end{sidewaystable}

\clearpage
Table~\ref{tab:model_comparison} presents the parameter estimates and goodness-of-fit statistics for the competing models fitted to the data. Although NTLE estimated through MLE shows the smallest AIC value, the NTLE model estimated through the MGFE method yields the smallest Kolmogorov-Smirnov statistic and the largest $p$-value, indicating a closer agreement between the fitted distribution and the empirical data. In all, the results suggest that incorporating the additional flexibility of the NTLE distribution allows the model to capture the underlying structure of the data more effectively than the base alternatives.

\begin{figure}[H]
\centering
\begin{minipage}{0.48\textwidth}
    \centering
    \includegraphics[width=\linewidth]{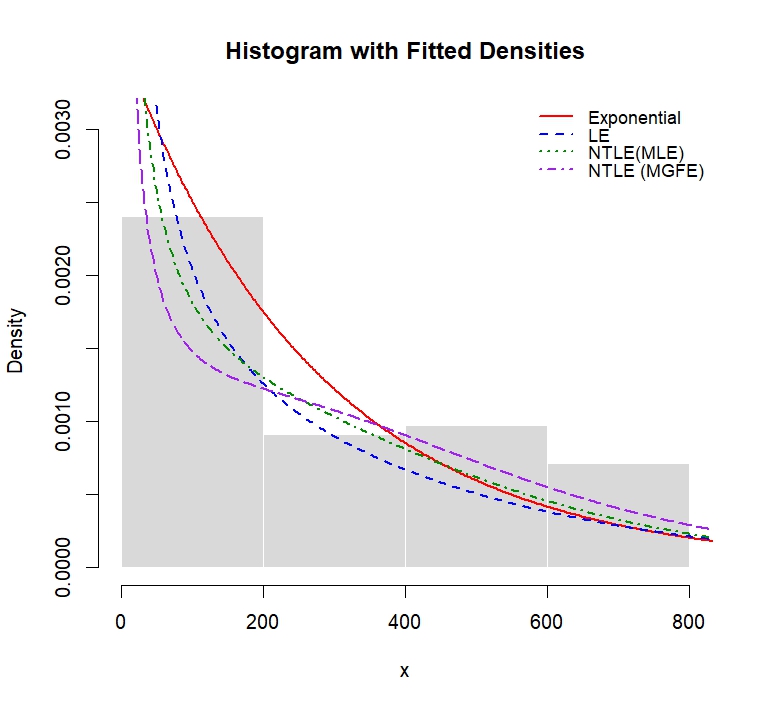}
    \caption{Histogram of COVID - 19 deaths in Egypt with fitted density functions of the competing models.}
    \label{hist}
\end{minipage}\hfill
\begin{minipage}{0.48\textwidth}
    \centering
    \includegraphics[width=\linewidth]{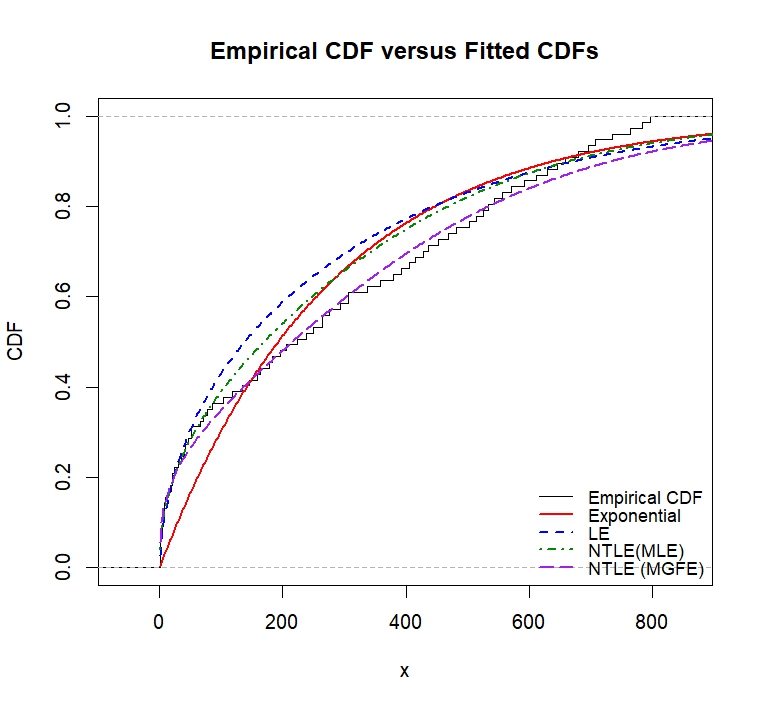}
\caption{Empirical cumulative distribution function and fitted cumulative distribution functions of the competing models.}
    \label{emp}
\end{minipage}
\end{figure}

The histogram with the fitted distributions is shown in Figure \ref{hist} while the plot of the empirical cumulative distribution function and fitted cumulative distribution functions of the models is given in Figure \ref{emp}. The visual comparisons shown in Figures~\ref{hist} and \ref{emp} further support the adequacy of the fitted NTLE model. The fitted NTLE(MGFE) density follows the shape of the histogram closely, while the fitted cumulative distribution function tracks the empirical CDF more accurately than other compared models.

\section{Discussion of Results}

The simulation results suggest that MGFE is a strong and reliable choice, especially in small and moderate samples, while MLE and a few alternative methods become increasingly competitive as the sample size grows. This findings also confirm that estimation of the transmutation parameter $\delta$ is generally more difficult than estimation of $\lambda$ and $\beta$, particularly in smaller samples.

In addition, this study also provides useful insight into the performance of the New Transmuted Logistic-exponential (NTLE) distribution when applied to real lifetime data. The NTLE model was fitted using the maximum goodness-of-fit estimation (MGFE) and MLE method, while the base models (Exponential and Logistic-Exponential) distributions were estimated using maximum likelihood procedures and compared using several information criteria and goodness-of-fit statistics (loglikelihood, AIC, BIC, CAIC, HQIC, and K-S test) 

The results reported in Table~\ref{tab:model_comparison} indicate that the NTLE model provides a competitive fit to the data when compared with the base distributions. In particular, the values of the information criteria and the K-S statistic suggest that the NTLE distribution is capable of modelling the observed data with a high degree of accuracy. The relatively small goodness-of-fit statistics demonstrate that the NTLE model successfully captures both the central tendency and the tail behaviour of the dataset.

The graphical diagnostics further support these findings. as shown in Figure~\ref{hist}, the fitted NTLE density curve closely follows the shape of the empirical histogram, indicating that the model accurately represents the distribution of the observed values. Likewise, the empirical cumulative distribution function plotted in Figure~\ref{emp} shows that the fitted NTLE cumulative distribution aligns closely with the empirical CDF across the entire range of the data. This visualisation supports the conclusions drawn from the numerical goodness-of-fit measures. Consequently, the NTLE model serves as a useful alternative to existing exponnetial and logistic-exponential distributions in practical reliability and lifetime data analysis.

\section{Conclusion}

This paper introduced and studied some new statistical properties of the NTLE distribution yet to be explored in literature, including Shannon entropy, R\'enyi entropy, stochastic ordering, mode, stress-strength reliability measure, residual life functions (mean and reverse), incomplete moments, Bonferroni and Lorenz curves. These theoretical developments provide a deeper understanding of the structural characteristics of NTLE and its potential applicability in practical data analysis. Comprehensive investigation of estimation techniques was also carried out. In particular, ten estimation methods were considered, including maximum likelihood, method of moments, least squares, weighted least, maximum product of spacings, Bayesian, Anderson–Darling, Cramér-von Mises, percentile estimation, and maximum goodness-of-fit estimation. A Monte Carlo simulation study was conducted to examine the performance of these estimators. The results of the simulation experiment indicated that the estimators exhibit different levels of accuracy depending on the estimation approach and the sample size, with likelihood-based and goodness-of-fit-based estimators generally providing stable and reliable results.

To demonstrate the practical relevance of NTLE distribution, an application using an infectious disease mortality data (COVID-19 deaths) was compared with its base models (Exponential and Logistic–Exponential). Model comparison was carried out using a variety of goodness-of-fit measures, including the log-likelihood value, AIC, BIC, CAIC, HQIC, and the Kolmogorov–Smirnov statistic. The empirical results, together with graphical diagnostics based on the fitted densities and cumulative distribution functions, showed that the NTLE distribution provides a competitive and flexible fit to the data.

Succinctly put, the findings of this study suggest that the NTLE distribution constitutes a useful addition to the class of logistic–exponential type models. Its relatively simple mathematical structure, combined with enhanced flexibility, makes it suitable for modelling a wide range of skewed data. Future research may explore additional inferential aspects of the distribution and applications to more complex datasets arising in modern-day fields.

\section*{Declarations}

\subsection*{Funding} No funding is received for this research. 

\subsection*{Conflict of interest} The authors declared no conflict of interest
 
\subsection*{Ethics approval and consent to participate} Not required

\subsection*{Data availability} The data used is secondary and cited appropriately in the manuscript

\bibliography{sn-bibliography}

\end{document}